\documentclass[10pt,a4paper]{article}
\usepackage[T1]{fontenc}
\usepackage[british]{babel}
\usepackage{amsmath,amsthm,amssymb,mathtools}
\usepackage{newpxtext,newpxmath}
\usepackage{bm}
\usepackage[a4paper,margin=26mm,headheight=15pt,headsep=8mm,footskip=11mm]{geometry}
\usepackage{microtype}
\usepackage{enumitem,booktabs,array}
\usepackage[most]{tcolorbox}
\usepackage{fancyhdr,titlesec,needspace}
\usepackage{tikz}
\usetikzlibrary{decorations.pathreplacing}
\usepackage{caption}
\usepackage[hidelinks,unicode,pdfencoding=auto]{hyperref}
\usepackage{bookmark}
\hypersetup{pdftitle={The generalised semi-Clifford conjecture is false},pdfauthor={Nadish de Silva and Oscar Lautsch},pdfsubject={Proofs, construction and exact circuit realisation}}
\numberwithin{equation}{section}
\titleformat{\section}{\large\bfseries}{\thesection}{0.7em}{}
\titleformat{\subsection}{\normalsize\bfseries}{\thesubsection}{0.7em}{}
\titleformat{\subsubsection}{\normalsize\itshape}{\thesubsubsection}{0.7em}{}
\titlespacing*{\section}{0pt}{21pt plus 3pt minus 2pt}{8pt}
\titlespacing*{\subsection}{0pt}{15pt plus 2pt minus 1pt}{5pt}
\titlespacing*{\subsubsection}{0pt}{11pt plus 2pt minus 1pt}{4pt}
\fancypagestyle{plain}{\fancyhf{}\fancyfoot[C]{\small\thepage}}
\setlist[enumerate,1]{label=\roman*),leftmargin=1.75em,itemsep=4pt,topsep=4pt,parsep=2pt}
\setlist[enumerate,2]{label=\alph*),leftmargin=1.75em,itemsep=2pt,topsep=3pt}
\newtheoremstyle{statement}{0pt}{0pt}{\normalfont}{} {\bfseries}{.}{0.5em}{}
\theoremstyle{statement}
\newtheorem{lemma}{Lemma}[section]
\newtheorem{theorem}[lemma]{Theorem}

\newtheorem{corollary}[lemma]{Corollary}
\newtheorem{definition}[lemma]{Definition}
\tcbset{statementbox/.style={enhanced,colback=black!2,colframe=black!2,boxrule=0pt,arc=0pt,outer arc=0pt,left=9pt,right=9pt,top=8pt,bottom=8pt,before skip=10pt,after skip=7pt,pad at break*=3pt}}
\tcolorboxenvironment{lemma}{statementbox}
\tcolorboxenvironment{theorem}{statementbox,colback=black!3}
\tcolorboxenvironment{proposition}{statementbox}
\tcolorboxenvironment{corollary}{statementbox}
\tcolorboxenvironment{definition}{statementbox}
\newtcolorbox{status}[1]{enhanced,colback=black!1,colframe=black!25,boxrule=0pt,borderline north={0.45pt}{0pt}{black!30},borderline south={0.35pt}{0pt}{black!20},arc=0pt,left=10pt,right=10pt,top=8pt,bottom=8pt,before skip=12pt,after skip=10pt,title={#1},fonttitle=\bfseries\small,coltitle=black,colbacktitle=black!5,attach title to upper={\par\smallskip}}

\newcommand{\HH}{\mathcal H}
\newcommand{\PP}{\mathcal P}
\newcommand{\CC}{\mathcal C}
\newcommand{\Alg}{\mathcal A}

\newcommand{\TT}{\mathbb T}
\newcommand{\ZZ}{\mathbb Z}
\newcommand{\Cplx}{\mathbb C}
\newcommand{\Ztwo}{\mathbb Z_{\mathtt{2}}}
\newcommand{\Cliff}{\operatorname{Cliff}}
\newcommand{\Sp}{\operatorname{Sp}}
\newcommand{\GL}{\operatorname{GL}}
\newcommand{\End}{\operatorname{End}}
\newcommand{\spanC}{\operatorname{span}_{\mathbb C}}
\newcommand{\spanZ}{\operatorname{span}_{\mathbb Z_{\mathtt{2}}}}
\newcommand{\im}{\operatorname{im}}
\newcommand{\Id}[1]{\mathbb{I}_{#1}}
\newcommand{\ctrl}{\operatorname{ctrl}}
\newcommand{\CNOT}{\operatorname{CX}}
\newcommand{\CZ}{\operatorname{CZ}}
\newcommand{\CCZ}{\operatorname{CCZ}}
\newcommand{\Toffoli}{\operatorname{CCX}}
\newcommand{\ket}[1]{\lvert #1\rangle}
\newcommand{\bra}[1]{\langle #1\rvert}
\newcommand{\proj}[1]{\ket{#1}\!\bra{#1}}
\newcommand{\vx}{\bm x}
\newcommand{\vz}{\bm z}
\newcommand{\vp}{\bm p}

\newcommand{\vu}{\bm u}
\newcommand{\vecv}{\bm v}
\newcommand{\vw}{\bm w}
\newcommand{\va}{\bm a}
\newcommand{\vb}{\bm b}
\newcommand{\ve}{\bm e}
\newcommand{\sA}{\mathsf A}
\newcommand{\sB}{\mathsf B}
\newcommand{\sC}{\mathsf C}
\newcommand{\sD}{\mathsf D}
\newcommand{\sE}{\mathsf E}
\newcommand{\sF}{\mathsf F}
\newcommand{\sJ}{\mathsf J}
\newcommand{\sK}{\mathsf K}
\newcommand{\sN}{\mathsf N}
\newcommand{\sR}{\mathsf R}
\newcommand{\Gstar}{\mathsf G_*}
\newcommand{\arxiv}[1]{\href{https://arxiv.org/abs/#1}{arXiv:#1}}

\begin{document}
\thispagestyle{plain}
\begin{center}
{\LARGE\bfseries The generalised semi-Clifford conjecture is false\par}
\vspace{10pt}
\begin{tabular}{@{}c@{\hspace{8em}}c@{}}
{\large Nadish de Silva} & {\large Oscar Lautsch} \\[2pt]
{\small Department of Mathematics} &
{\small Department of Pure Mathematics} \\
{\small Simon Fraser University} &
{\small Institute for Quantum Computing} \\
& {\small University of Waterloo}
\end{tabular}
\end{center}
\vspace{5pt}
\begin{abstract}
\noindent
The Clifford hierarchy is a nested sequence of sets of quantum gates that can be fault-tolerantly performed using gate teleportation within standard quantum error correction schemes.  The importance of these gates has motivated numerous studies of their structure.  Zeng--Chen--Chuang conjectured in 2007 that all hierarchy gates are \textit{generalised semi-Clifford}, i.e. take the form $C_1 \Pi D C_2$ for Clifford gates $C_1, C_2$, a permutation gate $\Pi$, and a diagonal gate $D$; Beigi--Shor proved in 2008 that this holds for all third-level gates.

We construct a five-qubit gate that is in the fifth level of the Clifford hierarchy but is not generalised semi-Clifford.  Rather than simply present and verify our counterexample to the generalised semi-Clifford conjecture, we show how its form can be deduced.  Our counterexample also demonstrates that the Clifford hierarchy is not closed under inverses.
\end{abstract}

\section{Introduction}

The Clifford hierarchy is a nested sequence of sets of quantum gates, introduced by Gottesman--Chuang in 1999, that is critical to fault-tolerant quantum computation~\cite{GC}. The first two levels are the Pauli and Clifford groups; higher levels contain the non-Clifford gates required for fault-tolerant universality via gate teleportation using magic states. Despite extensive work~\cite{ZCC,BS,CGK,RCP,PRTC,deSilva,ChenDeSilva,Anderson,deSilvaLautsch,AW,SDK,XW,HRT,AC,BGJ,BGPSW,BBCDH,Jamneshan,BBCH}, a general algebraic classification of the hierarchy remains open.

Generalised semi-Clifford (GSC) gates were introduced by Zeng--Chen--Chuang~\cite{ZCC} as a proposed structural description of the hierarchy. They are precisely the unitaries of the form $U=C_1\Pi D C_2$ where \(C_1,C_2\) are Clifford gates, \(\Pi\) is a permutation gate, and \(D\) is diagonal. They conjectured that every hierarchy gate has this form~\cite{ZCC}. Beigi--Shor proved this at the third level for arbitrarily many qubits~\cite{BS}.

We show in Theorem~\ref{thm:counterexample} that the following five-qubit gate, a two-control multiplexed-Clifford gate, is in the fifth level of the Clifford hierarchy but is not generalised semi-Clifford.  Tensoring with identities yields non-generalised-semi-Clifford
gates in the $k$-th level of the $n$-qubit Clifford hierarchy for every $n\ge5$ and $k\ge5$.

\begin{figure}[htbp]
\centering
\begin{tikzpicture}[x=1.24cm,y=0.76cm,line width=0.55pt]
  \foreach \y/\lab in {4/c,3/d,2/1,1/2,0/3} {
    \draw (0,\y)--(6.3,\y);
    \node[anchor=east] at (-0.14,\y) {$\mathtt{\lab}$};
  }
  \draw (0.85,0)--(0.85,3);
  \fill (0.85,3) circle (1.65pt);
  \fill (0.85,1) circle (1.65pt);
  \node[draw,fill=white,minimum width=6mm,minimum height=6mm,inner sep=1pt] at (0.85,0) {$Z$};
  \draw (1.95,0)--(1.95,3);
  \fill (1.95,3) circle (1.65pt);
  \fill (1.95,1) circle (1.65pt);
  \draw[fill=white] (1.95,0) circle (3.8pt);
  \draw (1.95,-0.176)--(1.95,0.176);
  \draw (1.842,0)--(2.058,0);
  \draw (3.15,2)--(3.15,4);
  \fill (3.15,4) circle (1.65pt);
  \node[draw,fill=white,minimum width=6mm,minimum height=6mm,inner sep=1pt] at (3.15,2) {$H$};
  \draw (4.3,1)--(4.3,4);
  \fill (4.3,4) circle (1.65pt);
  \fill (4.3,2) circle (1.65pt);
  \draw[fill=white] (4.3,1) circle (3.8pt);
  \draw (4.3,0.824)--(4.3,1.176);
  \draw (4.192,1)--(4.408,1);
  \draw (5.45,0)--(5.45,4);
  \fill (5.45,4) circle (1.65pt);
  \fill (5.45,1) circle (1.65pt);
  \draw[fill=white] (5.45,0) circle (3.8pt);
  \draw (5.45,-0.176)--(5.45,0.176);
  \draw (5.342,0)--(5.558,0);
  \draw[decorate,decoration={brace,amplitude=4pt}] (0.53,4.46)--(2.27,4.46)
    node[midway,above=6pt] {$\ctrl_{\mathtt{d}}(A)$};
  \draw[decorate,decoration={brace,amplitude=4pt}] (2.83,4.46)--(5.77,4.46)
    node[midway,above=6pt] {$\ctrl_{\mathtt{c}}(B)$};
\end{tikzpicture}
\caption{The five-qubit counterexample as a circuit: the first two gates implement the controlled-Clifford gate $\ctrl_{\mathtt{d}}(A)$; the last three implement the controlled-Clifford gate $\ctrl_{\mathtt{c}}(B)$.}
\label{fig:circuit}
\end{figure}
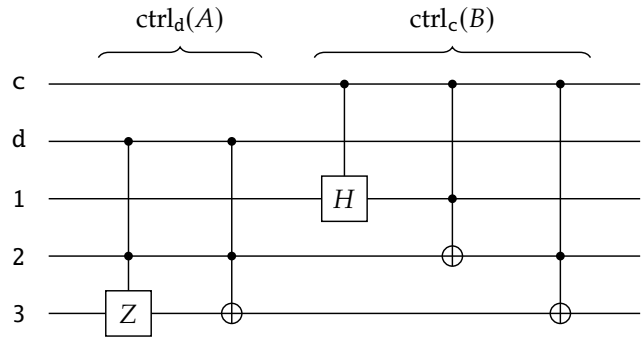

Moreover, we show in Theorem \ref{prop:inverse-outside-hierarchy} that the inverse of this gate does not belong to any level of the Clifford hierarchy.  This resolves another open structural question about the hierarchy.

We begin from the presumption that the conjecture is false, justified by the intuition that the two notions impose different kinds of structure: belonging to the Clifford hierarchy constrains how a gate conjugates Pauli operators recursively, whereas generalised semi-Cliffordness requires it to map one maximal commuting Pauli algebra onto another. There is no evident reason why recursive simplification should enforce this global compatibility. Controlled-Clifford gates provide a natural ansatz: existing criteria due to Anderson–-Weippert~\cite{AW} and Xu--Wang~\cite{XW}, building on work of Surti-–Daguerre-–Kim~\cite{SDK}, characterise hierarchy membership in the case of singly-controlled-Clifford gates.  Further, using Clifford target gates allows us to characterise generalised semi-Cliffordness and deploy symplectic linear algebra. We show below that one control qubit cannot provide a counterexample, so we study two-control multiplexed-Clifford gates. Each section below reduces the mathematical objects (and the conditions they must satisfy) that remains to be constructed, until solving explicit equations produces a gate guaranteed to be a counterexample.

\begin{itemize}

\item \textbf{\hyperref[sec:controlled]{Section~\ref*{sec:controlled}} gives characterisations of precisely which multiplexed-Clifford gates belong to the Clifford hierarchy and which are GSC; these are used to show that all singly-controlled-Clifford gates in the hierarchy are GSC.}
The gates under consideration are in the Clifford hierarchy when their commutators with a reduced family of Pauli gates are.  They are GSC when the symplectic matrices associated to the target Clifford gates have a common invariant Lagrangian subspace. Combining this with the aforementioned singly-controlled-Clifford hierarchy criterion~\cite{XW,AW,SDK} forces us to consider at least two controls.

\item \textbf{\hyperref[sec:two-control]{Section~\ref*{sec:two-control}} develops a sufficient condition for a large family of two-control multiplexed-Clifford gates to belong to the hierarchy.}
We determine conditions on the conditional target Clifford gates that allow us to conclude that they are hierarchy gates.  This is achieved by choosing the target Clifford gates so that their commutators with Pauli gates are effectively singly-controlled Clifford gates; this allows us to check when they, and thus the original gate, are in the hierarchy.  We find that the fourth-level case of this criterion always gives GSC gates, so we assume our counterexample lies at the fifth level or above.

\item \textbf{\hyperref[sec:clifford-pauli]{Section~\ref*{sec:clifford-pauli}} reduces the search for a suitable family of four target Clifford gates to one Clifford \(B\) and one Pauli \(P\) satisfying certain conditions.}
For our construction to be a hierarchy gate using the test we have developed, we find that the four Clifford target gates are of the form $B^{z_\mathtt{c}} A^{z_\mathtt{d}}$ for control bit values $z_{\mathtt c}, z_{\mathtt d}$.  We simplify our search for $A,B$ by imposing commutation relations and find that $A$ can be built from $B$ and a $\frac{\pi}{2}$-rotation $J$ about a Pauli $P$.

\item \textbf{\hyperref[sec:cyclic]{Section~\ref*{sec:cyclic}} shows that a suitable cyclic basis of a symplectic vector space yields a fifth-level hierarchy gate that is guaranteed to not be GSC and is excluded from the fourth level.}
A basis generated by successive applications of a nilpotent matrix associated with the symplectic matrix associated to the Clifford gate $B$ forces the failure of the GSC criterion.

\item \textbf{\hyperref[sec:symplectic-realisation]{Section~\ref*{sec:symplectic-realisation}} solves the equations specifying that configuration.}
We realise these data with explicit three-qubit Paulis, leaving only a real Clifford gate with the prescribed action to construct.

\item \textbf{\hyperref[sec:exact-realisation]{Section~\ref*{sec:exact-realisation}} converts the prescribed action into an exact circuit.}
Standard Pauli-conjugation rules produce the required real Clifford and determine the other target Clifford; adding the controls gives the explicit five-qubit gate.

\item \textbf{\hyperref[sec:counterexample]{Section~\ref*{sec:counterexample}} assembles the conclusion.}
The circuit satisfies the previously established construction conditions, so it lies in \(\mathcal{C}_5\setminus\mathcal{C}_4\) and is not generalised semi-Clifford.  We show using our techniques that its inverse is not in the Clifford hierarchy.

\item Our research process is detailed in the Section on
\hyperref[sec:AI]{\textbf{Use of  large language models}}.

\end{itemize}

\subsection{Conventions and background}
Throughout, $n$ denotes an integer with $n\geq 1$, unless otherwise specified.  We specify the phase group and the symplectic representation, and recall the algebraic definition of the Clifford hierarchy and generalised semi-Clifford gates.

Let $\HH_n=(\Cplx^2)^{\otimes n}$, and let $\mathrm U(\HH_n)$ denote its unitary group. The group of scalar phases is
\begin{equation}
\TT:=\{t\in\Cplx:|t|=1\}.
\end{equation}
With $\Ztwo=\ZZ/2\ZZ$, $Y=iXZ$, $X^{\vu} :=\bigotimes_{j=1}^n X^{u_j}$, and $Z^{\vecv}:=\bigotimes_{j=1}^n Z^{v_j}$, the discrete-phase Pauli group and the Clifford group are defined by
\begin{align}
\PP_n&:=\{i^\ell X^{\vu}Z^{\vecv}:(\ell,\vu,\vecv)\in(\ZZ/4\ZZ)\times\Ztwo^n\times\Ztwo^n\}\subseteq\mathrm U(\HH_n),\label{eq:finite-pauli}\\
\Cliff_n&:=\{C\in\mathrm U(\HH_n):C\PP_nC^\dagger=\PP_n\}.
\end{align}
An operator on specified registers acts by the identity on  all unnamed registers. For a vector space $E$, the identity is $\Id E$.  We identify $t\in\TT$ with $t\Id{\HH_n}$ when it multiplies an operator. Qubit labels are written in typewriter font; $\HH_{\mathtt{j}}=\Cplx^2$ denotes the Hilbert space of qubit $\mathtt{j}$.

The Pauli quotient is the $2n$-dimensional $\Ztwo$-vector space
\begin{equation}
V_n:=\PP_n/\{\pm\Id{\HH_n},\pm i\Id{\HH_n}\}\cong\Ztwo^n\oplus\Ztwo^n,
\qquad [X^{\vu}Z^{\vecv}]=(\vu,\vecv).
\label{eq:pauli-space}
\end{equation}
We define a symplectic form $\omega:V_n\times V_n\to\Ztwo$ by
\begin{equation}
\omega\bigl((\vu,\vecv),(\vu',\vecv')\bigr)=\vu\cdot\vecv'+\vecv\cdot\vu'.
\end{equation}
For all $P,Q\in\PP_n$,
\begin{equation}
[PQ]=[P]+[Q],\qquad PQ=(-1)^{\omega([P],[Q])}QP.
\label{eq:commutation}
\end{equation}
A subspace $L\leq V_n$ is Lagrangian precisely when $\dim_{\Ztwo}L=n$ and $\omega|_{L\times L}=0$. Its Pauli algebra is
\begin{equation}
\Alg_L:=\spanC\{P\in\PP_n:[P]\in L\}\subseteq\End_{\Cplx}(\HH_n).
\label{eq:pauli-algebra}
\end{equation}
We call such an algebra a \emph{maximal commuting Pauli algebra}.

A linear operation on $V_n$ is called symplectic if it preserves $\omega$, and we denote the group of invertible symplectic operators by $\Sp(V_n,\omega)$.
There is a symplectic representation of the Clifford group given by
\begin{equation}
\Cliff_n\longrightarrow\Sp(V_n,\omega),\qquad C\longmapsto\sC,
\end{equation}
where we define 
\begin{equation}
\sC[P]=[CPC^\dagger] \quad\text{for all }P\in\PP_n. 
\label{eq:representation}
\end{equation}
We call $\sC$ the \emph{symplectic matrix corresponding to the Clifford gate $C$}; the same letter with these two fonts denotes corresponding pairs. The Pauli coordinates in \eqref{eq:pauli-space} identify it with a $2n\times2n$ matrix over $\Ztwo$.

Sans-serif notation always refers to linear operators on vector spaces over $\Ztwo$. Such operators need not be symplectic or invertible: for example, we will frequently consider matrices of the form
\[
\sN=\sB-\Id{V_n}\in\End_{\Ztwo}(V_n)
\]
with $\sB\in \Sp(V_n, \omega)$.

Gottesman and Chuang introduced the Clifford hierarchy~\cite{GC} in 1999.
\begin{definition}[Clifford hierarchy]\label{def:hierarchy}
The first level is $\CC_1(n):=\TT\PP_n$. For all integers $k\geq1$, define
\begin{equation}
\CC_{k+1}(n):=\{U\in\mathrm U(\HH_n):\text{for all }P\in\PP_n,\ UPU^\dagger\in\CC_k(n)\}.
\label{eq:hierarchy}
\end{equation}
\end{definition}

The definition of generalised semi-Clifford gates and the conjecture that every hierarchy gate has this property are due to Zeng, Chen, and Chuang~\cite{ZCC}.
\begin{definition}[Generalised semi-Clifford gate]\label{def:gsc}
A unitary $U\in\mathrm U(\HH_n)$ is \emph{generalised semi-Clifford} (GSC) if there exist maximal abelian subgroups $G,G'\leq\PP_n$ such that
\begin{equation}
U(\spanC G)U^\dagger=\spanC G'.
\label{eq:gsc-subgroups}
\end{equation}
\end{definition}
Equivalently, there exist Lagrangians $L,L'\leq V_n$ such that
\begin{equation}
U\Alg_LU^\dagger=\Alg_{L'}.
\label{eq:algebraic-gsc}
\end{equation}
Zeng, Chen, and Chuang~\cite[Proposition~2]{ZCC} also give the equivalent form
\begin{equation}
U=C_1\Pi D C_2,
\label{eq:gsc-normal-form}
\end{equation}
where $C_1,C_2\in\Cliff_n$, $\Pi\in\mathrm U(\HH_n)$ is an arbitrary permutation matrix in the computational basis, and $D\in\mathrm U(\HH_n)$ is diagonal in that basis. 

\begin{lemma}[The symplectic representation and Pauli algebras]\label{lem:symplectic}\leavevmode
\begin{enumerate}
\item The representation $C\mapsto\sC$ is a group homomorphism with kernel $\TT\PP_n$, and $\CC_2(n)=\Cliff_n$.
\item The maximal abelian subgroups of $\PP_n$ are precisely
\[
\{P\in\PP_n:[P]\in L\},\qquad L\leq V_n\text{ Lagrangian}.
\]
Their complex spans are the algebras $\Alg_L$, which determine $L$ uniquely. Each $\Alg_L$ is Clifford-conjugate to the full computational-basis diagonal algebra.
\item For all $C\in\Cliff_n$ and all Lagrangians $L\leq V_n$,
\begin{equation}
C\Alg_LC^\dagger=\Alg_{\sC(L)}.
\label{eq:algebra-conjugation}
\end{equation}
\end{enumerate}
\end{lemma}
These are standard properties of the Pauli and Clifford groups; see~\cite{ZCC,BS,NRS}.

\begin{lemma}[Hierarchy invariances]\label{lem:invariances}
\leavevmode
\begin{enumerate}
\item The levels of the hierarchy are nested. Each level is invariant under multiplication by $\TT$, left and right Pauli multiplication, and Clifford conjugation.
\item For all integers $k\geq2$, all $U\in\mathrm U(\HH_n)$, and all $C_1,C_2\in\Cliff_n$,
\begin{equation}
U\in\CC_k(n)\quad\Longleftrightarrow\quad C_1UC_2\in\CC_k(n).
\end{equation}
\item For all $U\in\mathrm U(\HH_n)$ and all $C_1,C_2\in\Cliff_n$,
\[
U\text{ is GSC}\quad\Longleftrightarrow\quad C_1UC_2\text{ is GSC}.
\]
\item For all integers $r\geq1$ and all $C\in\Cliff_n$, $\Id{\HH_r}\otimes C\in\Cliff_{r+n}$. Conjugation by a permutation of qubit tensor factors preserves the Pauli group and every hierarchy level.
\end{enumerate}
\end{lemma}
\begin{proof}
See the hierarchy and Clifford invariances in~\cite{GC,deSilvaLautsch} and the Clifford-equivalence characterisation of GSC in~\cite{ZCC}.
\end{proof}

\section{Multiplexed-Clifford gates}
\label{sec:controlled}
Let $r\geq1$ be an integer. For a family $(C_{\vz})_{\vz\in\Ztwo^r}$ in $\mathrm U(\HH_n)$, denote its \textit{multiplexed gate} by
\begin{equation}
W:=\sum_{\vz\in\Ztwo^r}\proj{\vz}\otimes C_{\vz}\in\mathrm U(\HH_{r+n}).
\label{eq:controlled-family}
\end{equation}
The vector $\vz$ is a control value, and $C_{\vz}$ is its \textit{conditional target gate}. If all $C_{\vz}$ are Clifford, we call $W$ a \emph{multiplexed-Clifford gate}. For a separate control qubit $\mathtt{c}$ and $K\in\mathrm U(\HH_n)$, denote the singly-controlled gate by
\begin{equation}
\ctrl_{\mathtt{c}}(K):=\proj0_{\mathtt{c}}\otimes\Id{\HH_n}+\proj1_{\mathtt{c}}\otimes K.
\label{eq:one-control}
\end{equation}

We formulate useful equivalent characterisations of membership in the Clifford hierarchy and the generalised semi-Clifford property for multiplexed-Clifford gates.  Using the characterisation of singly-controlled Clifford gates \cite{XW, AW, SDK}, we show that any singly-controlled Clifford gate in the hierarchy is GSC.

\subsection{Criteria for hierarchy membership and the GSC property}
We characterise  whether a multiplexed-Clifford gate is in the hierarchy using Pauli commutators and characterise whether it is GSC by whether the symplectic matrices corresponding to its target Clifford gates share a common invariant Lagrangian.

\subsubsection{Pauli commutators and hierarchy membership}
We first state the commutator criterion for arbitrary gates.
For a controlled gate $W$, the $Z$-components of a Pauli $P$ on the control qubits do not affect the commutator between $W$ and $P$.

\begin{definition}[Pauli commutator]\label{def:commutator}
For all $(U,P)\in\mathrm U(\HH_n)\times\PP_n$, the Pauli commutator is
\begin{equation}
\Delta_PU:=PUP^\dagger U^\dagger\in\mathrm U(\HH_n).
\end{equation}
For all integers $r\geq1$, all $W\in\mathrm U(\HH_{r+n})$, and all $(\vx,Q)\in\Ztwo^r\times\PP_n$, denote
\begin{equation}
\Delta_{\vx,Q}W:=\Delta_{X^{\vx}\otimes Q}W.
\end{equation}
Here $\vx$ denotes a displacement of the control value $\vz$.
\end{definition}

\begin{lemma}[Hierarchy membership through Pauli commutators]\label{lem:commutator-test}
Let $k\geq1$ be an integer.
\begin{enumerate}
\item For all $U\in\mathrm U(\HH_n)$,
\[
U\in\CC_{k+1}(n)\quad\Longleftrightarrow\quad
\bigl[\text{for all }P\in\PP_n,\ \Delta_PU\in\CC_k(n)\bigr].
\]
\item Let $r\geq1$ and $(C_{\vz})_{\vz\in\Ztwo^r}$ be a family in $\mathrm U(\HH_n)$. For $W=\sum_{\vz}\proj{\vz}\otimes C_{\vz}$,
\begin{equation}
W\in\CC_{k+1}(r+n)\quad\Longleftrightarrow\quad
\bigl[\text{for all }(\vx,Q)\in\Ztwo^r\times\PP_n,\ \Delta_{\vx,Q}W\in\CC_k(r+n)\bigr].
\label{eq:controlled-test}
\end{equation}
\end{enumerate}
\end{lemma}
\begin{proof}
For (i), $\Delta_PU=P(UP^\dagger U^\dagger)$ lies in $\CC_k(n)$ exactly when $UP^\dagger U^\dagger$ does, by Pauli multiplication invariance. As $P^\dagger$ ranges over $\PP_n$, these are the defining conditions for $U\in\CC_{k+1}(n)$.

For (ii), write a full-register Pauli as $i^\ell(Z^{\va}X^{\vx}\otimes Q)$ with $\va, \vx\in \Ztwo^r, \ell\in \ZZ/4\ZZ$. The scalar cancels in conjugation; $X^{\vx}$ changes $\proj{\vz}$ to $\proj{\vz+\vx}$, and conjugation by $Z^{\va}$ fixes these projectors. Hence its commutator with $W$ is $\Delta_{\vx,Q}W$. Apply (i).
\end{proof}

\begin{lemma}[All Pauli commutators of a controlled gate]\label{lem:all-commutators}
Let $r\geq1$ be an integer, let $(C_{\vz})_{\vz\in\Ztwo^r}$ be a family in $\mathrm U(\HH_n)$, and let $W=\sum_{\vz}\proj{\vz}\otimes C_{\vz}$.
\begin{enumerate}
\item For all $(\vx,Q)\in\Ztwo^r\times\PP_n$ and all $\vz\in\Ztwo^r$, denote the \emph{conditional comparison gate} by
\begin{equation}
D_{\vx,Q}(\vz):=QC_{\vz+\vx}Q^\dagger C_{\vz}^\dagger\in\mathrm U(\HH_n).
\label{eq:comparison-operator}
\end{equation}
Then
\begin{equation}
\Delta_{\vx,Q}W=\sum_{\vz\in\Ztwo^r}\proj{\vz}\otimes D_{\vx,Q}(\vz).
\label{eq:all-commutators}
\end{equation}
\item If all $C_{\vz}$ are Clifford gates, then, for all $(\vx,Q)\in\Ztwo^r\times\PP_n$ and all $\vz\in\Ztwo^r$, $D_{\vx,Q}(\vz)\in\Cliff_n$ and its corresponding symplectic matrix is
\begin{equation}
\sD_{\vx}(\vz)=\sC_{\vz+\vx}\sC_{\vz}^{-1},
\label{eq:comparison-matrix-short}
\end{equation}
independently of $Q$.
\end{enumerate}
\end{lemma}
\begin{proof}
Every full-register Pauli has the form $i^\ell(Z^{\va}X^{\vx}\otimes Q)$, whose commutator was reduced to $(\vx,Q)$ above. Reindexing the control sum gives
\[
(X^{\vx}\otimes Q)W(X^{\vx}\otimes Q)^\dagger
=\sum_{\vz}\proj{\vz}\otimes QC_{\vz+\vx}Q^\dagger;
\]
multiplying by $W^\dagger$ proves (i). For (ii), $D_{\vx,Q}(\vz)$ is a product of Clifford gates, and the homomorphism in Lemma~\ref{lem:symplectic} sends both $Q$ and $Q^\dagger$ to identity, giving \eqref{eq:comparison-matrix-short}.
\end{proof}

\subsubsection{GSC in terms of a common invariant Lagrangian}
We use the algebraic characterisation of GSC due to Zeng, Chen, and Chuang~\cite{ZCC}, expressed in terms of the Pauli algebras $\Alg_L$. After setting one of the target operators in a multiplexed-Clifford gate $W$ to identity, that characterisation implies that $W$ is GSC if and only if the symplectic matrices corresponding to the target Clifford gates have a shared invariant Lagrangian.

For a Clifford family $(C_{\vz})_{\vz\in\Ztwo^r}$, choose $\vz_0\in\Ztwo^r$. We may normalise this family by replacing $C_{\vz}$ by $\widetilde C_{\vz}:=C_{\vz}C_{\vz_0}^\dagger$ and $W$ by
\begin{equation}
\widetilde W:=W(\Id{\HH_r}\otimes C_{\vz_0}^\dagger)
=\sum_{\vz}\proj{\vz}\otimes\widetilde C_{\vz}.
\label{eq:normalisation}
\end{equation}
Then $\widetilde C_{\vz_0}=\Id{\HH_n}$, and for all $\vz,\vx\in\Ztwo^r$,
\[
\widetilde C_{\vz+\vx}\widetilde C_{\vz}^{\dagger}=C_{\vz+\vx}C_{\vz}^{\dagger}.
\]
The right factor in \eqref{eq:normalisation} is Clifford, so Lemma~\ref{lem:invariances} shows that normalisation preserves GSC and, for all integers $k\geq2$, membership in $\CC_k(r+n)$. Products and adjoints of real matrices are real, so normalisation also preserves reality.

A maximal commuting Pauli algebra on the full register need not be a tensor product of control and target algebras. The next lemma isolates the target algebra that survives fixing a computational control value. Furthermore, it shows that this target algebra is a maximal commuting Pauli algebra on the target qubits.

\begin{lemma}[Compression of a Pauli algebra]\label{lem:compression}
Let $r\ge1$ be an integer, and let $\mathcal A$ be a maximal
commuting Pauli algebra on $\mathcal H_r\otimes\mathcal H_n$.
For each $\boldsymbol z\in\mathbb Z_2^r$, define the \emph{$\vz$-compressed space}
\[
\langle\boldsymbol z|\mathcal A|\boldsymbol z\rangle
:=
\left\{
(\langle\boldsymbol z|\otimes\mathbb I_{\mathcal H_n})
M
(|\boldsymbol z\rangle\otimes\mathbb I_{\mathcal H_n})
:\ M\in\mathcal A
\right\}.
\]
Then all these spaces coincide: there is a maximal commuting
Pauli algebra $\mathcal A_{L_{\mathrm{tar}}}$ on $\mathcal H_n$ such that
\[
\langle\boldsymbol z|\mathcal A|\boldsymbol z\rangle
=\mathcal \mathcal \mathcal A_{L_{\mathrm{tar}}}
\qquad\text{for every }\boldsymbol z\in\mathbb Z_2^r.
\]
\end{lemma}

\begin{proof}
Use the natural symplectic identification
$V_{r+n}=V_r\oplus V_n$.
Let $L\le V_r\oplus V_n$ be the Lagrangian corresponding to
$\mathcal A$, and let $L_Z$ be the Lagrangian corresponding to the diagonal algebra on the control qubits:
\[
L_Z:=\{[Z^{\boldsymbol b}]:\boldsymbol b\in\mathbb Z_2^r\}
\le V_r.
\]
Consider the surjective linear map
\[
\pi:L_Z\oplus V_n\longrightarrow V_n,
\qquad
(\boldsymbol u,\boldsymbol v)\longmapsto\boldsymbol v,
\]
and define
\[
L_{\mathrm{tar}}
:=\pi\bigl(L\cap(L_Z\oplus V_n)\bigr).
\]

For $\boldsymbol a,\boldsymbol b\in\mathbb Z_2^r$ and
$P\in\mathcal P_n$,
\[
\langle\boldsymbol z|
(X^{\boldsymbol a}Z^{\boldsymbol b}\otimes P)
|\boldsymbol z\rangle
=
\delta_{\boldsymbol a,\boldsymbol0}
(-1)^{\boldsymbol b\cdot\boldsymbol z}P.
\]
Applying this identity to a Pauli basis of $\mathcal A$ gives
\[
\langle\boldsymbol z|\mathcal A|\boldsymbol z\rangle
=
\operatorname{span}_{\mathbb C}
\{P\in\mathcal P_n:[P]\in L_{\mathrm{tar}}\}
\qquad
\text{for every }\boldsymbol z\in\mathbb Z_2^r.
\]
In particular, the compressed spaces are independent of
$\boldsymbol z$.

It remains to show that $L_{\mathrm{tar}}$ is Lagrangian.
Using $L^\perp=L$ in $V_r\oplus V_n$,
$L_Z^\perp=L_Z$ in $V_r$, and
$(E\cap F)^\perp=E^\perp+F^\perp$, we obtain, for every
$\boldsymbol v\in V_n$,
\[
\begin{aligned}
\boldsymbol v\in L_{\mathrm{tar}}^\perp
&\iff
(\boldsymbol0,\boldsymbol v)
\in\bigl(L\cap(L_Z\oplus V_n)\bigr)^\perp\\
&\iff
(\boldsymbol0,\boldsymbol v)
\in L+(L_Z\oplus\{\boldsymbol0\})\\
&\iff
\text{there exists }\boldsymbol u\in L_Z
\text{ such that }(\boldsymbol u,\boldsymbol v)\in L\\
&\iff
\boldsymbol v\in L_{\mathrm{tar}}.
\end{aligned}
\]
Hence $L_{\mathrm{tar}}=L_{\mathrm{tar}}^\perp$, so
$L_{\mathrm{tar}}$ is Lagrangian. The common compressed space is
therefore the maximal commuting Pauli algebra
$\mathcal A_{L_{\mathrm{tar}}}$.
\end{proof}

\begin{theorem}[Common invariant Lagrangian]\label{thm:common-lagrangian}
Let $r\geq1$ be an integer and let $(C_{\vz})_{\vz\in\Ztwo^r}$ be a family in $\Cliff_n$. Set
\[
W:=\sum_{\vz\in\Ztwo^r}\proj{\vz}\otimes C_{\vz}\in\mathrm U(\HH_{r+n}),
\]
and suppose $C_{\vz_0}=\Id{\HH_n}$ for some $\vz_0\in\Ztwo^r$. Then $W$ is GSC if and only if there exists a Lagrangian $L\leq V_n$ such that, for all $\vz\in\Ztwo^r$,
\begin{equation}
\sC_{\vz}(L)=L.
\label{eq:common-lagrangian}
\end{equation}
\end{theorem}
\begin{proof}
The controlled form of $W$ and orthogonality of the computational
basis give
\[
(\langle\boldsymbol z|\otimes\mathbb I_{\mathcal H_n})W
=
C_{\boldsymbol z}
(\langle\boldsymbol z|\otimes\mathbb I_{\mathcal H_n}),
\]
together with the adjoint identity for $W^\dagger$.

Suppose $W$ is GSC, so that
$W\mathcal A W^\dagger=\mathcal A'$ for maximal commuting
Pauli algebras $\mathcal A,\mathcal A'$.
By the compression lemma, their compressions are maximal
commuting target Pauli algebras $\mathcal B,\mathcal B'$
independent of $\boldsymbol z$. The identities above give
\[
\mathcal B'
=
\langle\boldsymbol z|W\mathcal A W^\dagger|\boldsymbol z\rangle
=
C_{\boldsymbol z}\mathcal B C_{\boldsymbol z}^\dagger
\qquad\text{for every }\boldsymbol z.
\]
Taking $\boldsymbol z=\boldsymbol z_0$ gives
$\mathcal B'=\mathcal B$.
Write $\mathcal B=\mathcal A_L$ for its corresponding
Lagrangian. Since
\[
C_{\boldsymbol z}\mathcal A_L C_{\boldsymbol z}^\dagger
=
\mathcal A_{\mathsf C_{\boldsymbol z}(L)},
\]
uniqueness of the corresponding Lagrangian implies
$\mathsf C_{\boldsymbol z}(L)=L$ for every $\boldsymbol z$.

Conversely, suppose such an $L$ exists, so every
$C_{\boldsymbol z}$ preserves $\mathcal A_L$.
Let $\mathcal D_r$ be the computational diagonal algebra
on the controls. The algebra $\mathcal D_r\otimes\mathcal A_L$
is maximal commuting Pauli, and
\[
W(\mathcal D_r\otimes\mathcal A_L)W^\dagger
=
\mathcal D_r\otimes\mathcal A_L:
\]
indeed, $\mathcal D_r$ is spanned by
$|\boldsymbol z\rangle\langle\boldsymbol z|$, and conjugation
by $W$ acts on the corresponding target factor by
$C_{\boldsymbol z}$. Hence $W$ is GSC.
\end{proof}
Thus after normalisation, hierarchy membership concerns the comparison operators $D_{\vx,Q}(\vz)$, whereas GSC membership concerns a common invariant Lagrangian for the normalised conditional matrices $\sC_{\vz}$.

\subsection{A single control}
The singly-controlled-Clifford hierarchy criterion will be used on each commutator. Anderson and Weippert~\cite{AW} established an earlier necessary power condition; Surti, Daguerre, and Kim~\cite{SDK} treated the third-level case. The following is the membership form of the controlled-Clifford criterion of Xu and Wang~\cite{XW}.

\begin{lemma}[Singly-controlled-Clifford hierarchy criterion; Xu--Wang]\label{lem:controlled-powers}
Let $m\geq0$ be an integer, let $K\in\Cliff_n$, and let $\mathtt{c}$ be a separate control qubit. Then
\begin{equation}
\ctrl_{\mathtt{c}}(K)\in\CC_{m+2}(n+1)\quad\Longleftrightarrow\quad K^{2^m}\in\PP_n.
\label{eq:controlled-powers}
\end{equation}
\end{lemma}

\begin{lemma}[Invariant Lagrangians and a one-dimensional kernel]\label{lem:nilpotent-invariance}
Let $\sB\in\Sp(V_n,\omega)$ and suppose $\sN:=\sB-\Id{V_n}\in\End_{\Ztwo}(V_n)$ is nilpotent.
\begin{enumerate}
\item There exists a Lagrangian $L\leq V_n$ satisfying $\sB(L)=L$.
\item If $\ker\sN=\Ztwo\vw$ for some $\vw\in V_n\setminus\{\bm0\}$, then, for all nonzero $\sB$-invariant subspaces $S\leq V_n$, $\vw\in S$.
\end{enumerate}
\end{lemma}
\begin{proof}
For (i), start with the invariant isotropic subspace $L=0$. If $L$ is $\sB$-invariant, then $\sB(L)=L$ by invertibility, and $L^\perp$ is also invariant: for all $\vecv\in L^\perp$ and all $\bm\ell\in L$,
\[
\omega(\sB\vecv,\bm\ell)=\omega(\vecv,\sB^{-1}\bm\ell)=0.
\]
Thus $L$ and $L^\perp$ are $\sN$-invariant. If $\dim L<n$, choose $\vecv\in L^\perp\setminus L$. Nilpotence gives a largest integer $a\geq0$ with $\sN^a\vecv\notin L$. Set $\vu:=\sN^a\vecv$. Then
\[
\vu\in L^\perp\setminus L,\qquad \sN\vu\in L.
\]
The space $L+\Ztwo\vu$ is isotropic and has dimension $\dim L+1$. It is $\sB$-invariant because $\sB\vu=\vu+\sN\vu$. Repeating this enlargement yields an invariant isotropic subspace of dimension $n$, hence a Lagrangian.

For (ii), let $S$ be a nonzero $\sB$-invariant subspace; it is also $\sN$-invariant. Choose $\vecv\in S\setminus\{\bm0\}$ and the largest $a\geq0$ with $\sN^a\vecv\ne\bm0$. Then $\sN^a\vecv$ is a nonzero element of $S\cap\ker\sN$, so it equals $\vw$.
\end{proof}

\begin{corollary}[One control]\label{cor:one-control}
Let $C_0,C_1\in\Cliff_n$ and set
\[
W:=\proj0_{\mathtt{c}}\otimes C_0+\proj1_{\mathtt{c}}\otimes C_1\in\mathrm U(\HH_{n+1}).
\]
If $W\in\CC_k(n+1)$ for some integer $k\geq1$, then $W$ is GSC.
\end{corollary}
\begin{proof}
By nesting, we may take $k\geq2$. Normalisation gives
\[
W(\Id{\HH_{\mathtt{c}}}\otimes C_0^\dagger)=\ctrl_{\mathtt{c}}(K),\qquad K:=C_1C_0^\dagger\in\Cliff_n,
\]
and preserves hierarchy membership and GSC. Lemma~\ref{lem:controlled-powers} gives $K^{2^{k-2}}\in\PP_n$, hence $\sK^{2^{k-2}}=\Id{V_n}$. In characteristic two,
\[
(\sK-\Id{V_n})^{2^{k-2}}=\sK^{2^{k-2}}-\Id{V_n}=0.
\]
Lemma~\ref{lem:nilpotent-invariance} then supplies a $\sK$-invariant Lagrangian $L$. As the symplectic matrices of the target Clifford gates of $\ctrl_{\mathtt{c}}(K)$ are $\Id{V_n}$ and $\sK$, both of which preserve $L$, Theorem~\ref{thm:common-lagrangian} makes the normalised gate GSC, and undoing normalisation proves the claim.
\end{proof}
Our construction therefore focuses on the two-controlled multiplexed-Clifford case. Beigi and Shor~\cite{BS} proved that every third-level qubit gate is GSC; thus we focus on potential counterexamples in level $k \geq 4$.

\begin{status}{Construction requirements after Section 2}
Find a family $(C_{\vz})_{\vz\in\Ztwo^2}$ in $\Cliff_n$, with $C_{\bm0}=\Id{\HH_n}$, such that
\[
W:=\sum_{\vz\in\Ztwo^2}\proj{\vz}\otimes C_{\vz}
\]
satisfies, for some $k\geq 4$  and all $(\vx,Q)\in\Ztwo^2\times\PP_n$,
\[
\Delta_{\vx,Q}W\in\CC_{k-1}(n+2),
\]
while the matrices $\sC_{\vz}$ have no common invariant Lagrangian. 
\end{status}

\section{A sufficient two-control criterion}
\label{sec:two-control}
We cannot apply the singly-controlled-Clifford hierarchy criterion directly to a two-controlled multiplexed-Clifford gate $W$. However, if each Pauli commutator of $W$ is a singly-controlled-Clifford, we may apply the criterion to the commutators to show that they are all in the hierarchy, and hence that $W$ is as well. The Pauli commutators of $W$ are naturally two-controlled, so we find sufficient conditions on the symplectic matrices of the targets of $W$ that allow us to interpret one control qubit as a target and express the commutators as singly-controlled-Clifford gates, up to a Clifford factor. 

\subsection{Comparison matrices and relative matrices}
We distinguish the symplectic matrices of the conditional comparison operators from the relative matrix describing their change between the two values of the first control.

The ordered registers are
\[
\HH_{n+2}=\HH_{\mathtt{c}}\otimes\HH_{\mathtt{d}}\otimes\HH_n,\qquad
\HH_{\mathtt{c}}=\HH_{\mathtt{d}}=\Cplx^2.
\]
For all $\vz=(z_{\mathtt{c}},z_{\mathtt{d}})\in\Ztwo^2$, write $\ket{\vz}_{\mathtt{c},\mathtt{d}}:=\ket{z_{\mathtt{c}}}_{\mathtt{c}}\otimes\ket{z_{\mathtt{d}}}_{\mathtt{d}}$. Pauli commutators are indexed jointly by $(\vx,Q)\in\Ztwo^2\times\PP_n$.

\begin{definition}[Comparison matrices]\label{def:comparison-matrices}
Let $(C_{\vz})_{\vz\in\Ztwo^2}$ be a family in $\Cliff_n$. For all $\vx\in\Ztwo^2$, define
\begin{equation}
\sD_{\vx}:\Ztwo^2\longrightarrow\Sp(V_n,\omega),\qquad
\sD_{\vx}(\vz):=\sC_{\vz+\vx}\sC_{\vz}^{-1}.
\label{eq:comparison-matrices}
\end{equation}
The values of $\sD_{\vx}$ are the \emph{comparison matrices}. For all $(\vx,Q)\in\Ztwo^2\times\PP_n$ and all $\vz\in\Ztwo^2$, the symplectic matrix corresponding to $D_{\vx,Q}(\vz)$ is $\sD_{\vx}(\vz)$, independently of $Q$.
\end{definition}

When $\vx$ is fixed, we write $\sD$ for $\sD_{\vx}$, and when $(\vx,Q)$ is fixed, we write $D$ for $D_{\vx,Q}$. We retain the indexed notation when comparing different values of $\vx$.

\begin{definition}[Relative matrix]\label{def:relative-matrix}
Let $(C_{\vz})_{\vz\in\Ztwo^2}$ be a family in $\Cliff_n$ and let $\vx\in\Ztwo^2$. Suppose that, for all $(a,b)\in\Ztwo^2$,
\begin{equation}
\sD(a,b)=\sD(a,0).
\label{eq:second-control-independence}
\end{equation}
The \emph{relative matrix} for $\vx$ is
\begin{equation}
\sR_{\vx}:=\sD(1,0)\sD(0,0)^{-1}\in\Sp(V_n,\omega).
\label{eq:relative-matrix}
\end{equation}
\end{definition}
As we will see in Lemmas \ref{lem:reduction} and \ref{lem:uniform-powers}, these relative matrices may be used to show that the Pauli commutators of the controlled-Clifford gate for this family are all in the hierarchy, under certain assumptions.
\subsection{Reduction to one control}
The second-control independence condition \eqref{eq:second-control-independence} is designed to let us include the second control $\mathtt{d}$ in a Clifford target register, leaving only $\mathtt{c}$ as a control. Indeed, if the comparison matrices are independent of the second control, then changing $\mathtt{d}$ changes the corresponding target operator only by a scalar phase and a discrete-phase Pauli.

A controlled discrete-phase Pauli is Clifford, but an arbitrary scalar phase need not have this property. To account for this, we restrict the conditional operators to real Clifford gates. Without reality, two Clifford gates with the same symplectic matrix can differ by any element of $\TT\PP_n$. Under a control, a scalar becomes a relative phase:
\[
\ctrl_{\mathtt{d}}(tK)=\bigl(\operatorname{diag}(1,t)_{\mathtt{d}}\otimes\Id{\HH_n}\bigr)\ctrl_{\mathtt{d}}(K).
\]
Reality reduces the scalar ambiguity to a sign relative to a real Pauli, so that the correction is exactly a discrete-phase Pauli.

\begin{lemma}[Real Clifford gates with the same symplectic matrix]\label{lem:real-kernel}
Let $E,F\in\Cliff_n$ be real and suppose $\sE=\sF$. Then
\begin{equation}
EF^\dagger\in\{\pm X^{\vu}Z^{\vecv}:(\vu,\vecv)\in\Ztwo^n\times\Ztwo^n\}\subseteq\PP_n.
\label{eq:real-kernel}
\end{equation}
In particular, a real Clifford whose corresponding symplectic matrix is $\Id{V_n}$ belongs to $\PP_n$.
\end{lemma}
\begin{proof}
By the kernel statement in Lemma~\ref{lem:symplectic}, $EF^\dagger=tX^{\vu}Z^{\vecv}$ for some $t\in\TT$. Since both $EF^\dagger$ and $X^{\vu}Z^{\vecv}$ are real, $t$ must be real, hence $t=\pm1$. Taking $F=\Id{\HH_n}$ gives the last assertion.
\end{proof}

\begin{lemma}[Reduction of a Pauli commutator]\label{lem:reduction}
Let $(C_{\vz})_{\vz\in\Ztwo^2}$ be a family of real Clifford gates on $\HH_n$, and let
\[
W:=\sum_{\vz\in\Ztwo^2}\proj{\vz}_{\mathtt{c},\mathtt{d}}\otimes C_{\vz}\in\mathrm U(\HH_{n+2}).
\]
Fix $(\vx,Q)\in\Ztwo^2\times\PP_n$, and set
\[
D(\vz):=QC_{\vz+\vx}Q^\dagger C_{\vz}^{\dagger}.
\]
Suppose that, for all $(a,b)\in\Ztwo^2$, the corresponding symplectic matrices satisfy
\[
\sD(a,b)=\sD(a,0).
\]
For all $a\in\Ztwo$, let
\[
F_a:=\sum_{b\in\Ztwo}\proj b_{\mathtt{d}}\otimes D(a,b),\qquad K:=F_1F_0^\dagger.
\]
Then the following statements hold.
\begin{enumerate}
\item $F_0,F_1,K\in\Cliff_{n+1}$ are real operators on $\HH_{\mathtt{d}}\otimes\HH_n$.
\item In the given tensor order,
\begin{equation}
\Delta_{\vx,Q}W=\ctrl_{\mathtt{c}}(K)(\Id{\HH_{\mathtt{c}}}\otimes F_0).
\label{eq:reduction}
\end{equation}
The remaining control is $\mathtt{c}$; its target register consists of the qubit $\mathtt{d}$ and the original $n$ targets.
\item For all $b\in\Ztwo$, let
\[
K_b:=D(1,b)D(0,b)^\dagger\in\Cliff_n.
\]
These operators are real and satisfy
\begin{equation}
K=\sum_{b\in\Ztwo}\proj b_{\mathtt{d}}\otimes K_b,\qquad
\sK_0=\sK_1=\sR_{\vx}:=\sD(1,0)\sD(0,0)^{-1}.
\label{eq:conditional-K}
\end{equation}
\end{enumerate}
\end{lemma}
\begin{proof}
Write $Q=i^\ell X^{\vu}Z^{\vecv}$. Its scalar cancels in $D(\vz)$, leaving a product of real matrices, so every comparison operator is real. Lemma~\ref{lem:real-kernel} and the assumed equality of symplectic matrices give
\[
P_a:=D(a,1)D(a,0)^\dagger\in\PP_n.
\]
Consequently,
\[
F_a=\ctrl_{\mathtt{d}}(P_a)\bigl(\Id{\HH_{\mathtt{d}}}\otimes D(a,0)\bigr)\in\Cliff_{n+1},
\]
by the $m=0$ case of Lemma~\ref{lem:controlled-powers}. The operators $F_a$ and $K=F_1F_0^\dagger$ are real, proving (i).

Grouping the control projectors by $a$ gives
\begin{align*}
\Delta_{\vx,Q}W
&=\proj0_{\mathtt{c}}\otimes F_0+\proj1_{\mathtt{c}}\otimes F_1\\
&=\bigl(\proj0_{\mathtt{c}}\otimes\Id{\HH_{\mathtt{d}}\otimes\HH_n}
+\proj1_{\mathtt{c}}\otimes F_1F_0^\dagger\bigr)(\Id{\HH_{\mathtt{c}}}\otimes F_0),
\end{align*}
which is (ii). Multiplying the two sums defining $F_1F_0^\dagger$ gives the formula for $K$ in (iii). The symplectic representation then gives
\[
\sK_b=\sD(1,b)\sD(0,b)^{-1}=\sR_{\vx}
\]
for all $b\in \Ztwo$.
\end{proof}
The Clifford $K$ acts on $n+1$ qubits, whereas its conditional target operators $K_0,K_1$ act on $n$ qubits. Consequently, $\sK\in\Sp(V_{n+1},\omega)$ is not the matrix $\sK_0=\sK_1=\sR_{\vx}\in\Sp(V_n,\omega)$.

In Lemma~\ref{lem:reduction}, Clifford multiplication invariance and Lemma~\ref{lem:controlled-powers}, applied to the $n+1$ target qubits, give, for all integers $m\geq0$,
\begin{equation}
\Delta_{\vx,Q}W\in\CC_{m+2}(n+2)
\quad\Longleftrightarrow\quad K^{2^m}\in\PP_{n+1}.
\label{eq:reduced-exact-test}
\end{equation}
In the next section we give sufficient conditions for $K^{2^m}$ to be Pauli for each choice of $(\vx, Q)$, which will imply that each commutator $\Delta_{\vx,Q}W$ is in the Clifford hierarchy. Corollary~\ref{cor:one-control} then makes its singly controlled factor in \eqref{eq:reduction} a GSC gate, and the surrounding Clifford factors preserve the GSC property. Thus every such commutator is GSC, but notably this does not imply that $W$ itself is GSC. 

\subsection{A sufficient power condition}
In the notation of Lemma \ref{lem:reduction}, the $2^m$th powers of the conditional target operators $K_0$ and $K_1$ must agree up to a sign in order for the $2^m$th power of $K$ to form a single Pauli on the enlarged register. The next lemma ensures this agreement from their reality and an additional nilpotence condition on their common symplectic matrix.

\begin{lemma}[Uniform Pauli powers]\label{lem:uniform-powers}
Let $K_0,K_1\in\Cliff_n$ be real, with corresponding symplectic matrices
\[
\sK_0=\sK_1=\sR\in\Sp(V_n,\omega).
\]
Set
\[
K:=\proj0_{\mathtt{d}}\otimes K_0+\proj1_{\mathtt{d}}\otimes K_1\in\mathrm U(\HH_{\mathtt{d}}\otimes\HH_n).
\]
Then $K\in\Cliff_{n+1}$. If $m\geq1$ is an integer and
\begin{equation}
(\sR-\Id{V_n})^{2^m-1}=0,
\label{eq:uniform-nilpotence}
\end{equation}
there exists $s\in\Ztwo$ such that
\begin{enumerate}
\item $K_0^{2^m}\in\PP_n$;
\item $K_1^{2^m}=(-1)^s K_0^{2^m}$;
\item $K^{2^m}=Z_{\mathtt{d}}^s\otimes K_0^{2^m}\in\PP_{n+1}$.
\end{enumerate}
\end{lemma}
\begin{proof}
By Lemma~\ref{lem:real-kernel}, $K_1=PK_0$ for a real discrete-phase Pauli $P$. Thus
\[
K=\ctrl_{\mathtt{d}}(P)(\Id{\HH_{\mathtt{d}}}\otimes K_0)\in\Cliff_{n+1}.
\]
Put $\nu=2^m$. In $\Ztwo[t]$,
\[
t^\nu-1=(t-1)^\nu,\qquad 1+t+\cdots+t^{\nu-1}=(t-1)^{\nu-1}.
\]
Substitution and \eqref{eq:uniform-nilpotence} give
\begin{equation}
\sR^\nu=\Id{V_n},\qquad
\sum_{j=0}^{\nu-1}\sR^j=(\sR-\Id{V_n})^{\nu-1}=0.
\label{eq:geometric-sum}
\end{equation}
The second polynomial identity is obtained by cancellation in the polynomial ring, not by dividing by the possibly singular matrix $\sR-\Id{V_n}$. The real Clifford $K_0^\nu$ has symplectic matrix identity, so Lemma~\ref{lem:real-kernel} gives (i).

Moving the factors $K_0$ to the right in $(PK_0)^\nu$ yields
\begin{equation}
K_1^\nu K_0^{-\nu}
=\prod_{j=0}^{\nu-1}K_0^j P K_0^{-j},
\label{eq:power-correction}
\end{equation}
with factors ordered by increasing $j$ from left to right. This is a real discrete-phase Pauli whose class is
\[
\sum_{j=0}^{\nu-1}\sR^j[P]=\bm0.
\]
It is therefore $(-1)^s\Id{\HH_n}$ for some $s\in \Ztwo$, proving (ii). Hence
\[
K^\nu=\proj0_{\mathtt{d}}\otimes K_0^\nu+\proj1_{\mathtt{d}}\otimes(-1)^sK_0^\nu
=Z_{\mathtt{d}}^s\otimes K_0^\nu,
\]
which proves (iii).
\end{proof}

\begin{theorem}[Two-control hierarchy criterion]\label{thm:two-control}
Let $m\geq1$ be an integer, let $(C_{\vz})_{\vz\in\Ztwo^2}$ be a family of real Clifford gates on $\HH_n$, and set
\[
W:=\sum_{\vz\in\Ztwo^2}\proj{\vz}_{\mathtt{c},\mathtt{d}}\otimes C_{\vz}\in\mathrm U(\HH_{n+2}).
\]
For all $\vx\in\Ztwo^2$, write
\[
\sD_{\vx}(\vz):=\sC_{\vz+\vx}\sC_{\vz}^{-1},\qquad
\sR_{\vx}:=\sD_{\vx}(1,0)\sD_{\vx}(0,0)^{-1}.
\]
Assume:
\begin{enumerate}
\item For all $\vx\in\Ztwo^2$ and all $(a,b)\in\Ztwo^2$, $\sD_{\vx}(a,b)=\sD_{\vx}(a,0)$;
\item For all $\vx\in\Ztwo^2$, $(\sR_{\vx}-\Id{V_n})^{2^m-1}=0$.
\end{enumerate}
Then
\begin{equation}
W\in\CC_{m+3}(n+2).
\label{eq:two-control-bound}
\end{equation}
\end{theorem}
\begin{proof}
Fix $(\vx,Q)\in\Ztwo^2\times\PP_n$. Lemma~\ref{lem:reduction} writes $\Delta_{\vx,Q}W=\ctrl_{\mathtt{c}}(K)(\Id{\HH_{\mathtt{c}}}\otimes F_0)$, with real $F_0,K\in\Cliff_{n+1}$. The conditional operators $K_0,K_1$ are real Clifford gates with common symplectic matrix $\sR_{\vx}$. Assumption (ii) and Lemma~\ref{lem:uniform-powers} give $K^{2^m}\in\PP_{n+1}$. Equation~\eqref{eq:reduced-exact-test} gives $\Delta_{\vx,Q}W\in\CC_{m+2}(n+2)$. Since this holds for all $(\vx,Q)$, Lemma~\ref{lem:commutator-test} proves the result.
\end{proof}
We will use this theorem to construct a gate in the hierarchy that is not GSC. As all third-level gates are GSC, the smallest possible hierarchy level of a non-GSC gate is 4, which corresponds to $m=1$ in this theorem.
However, the following lemma shows that any family of Clifford gates satisfying the hypotheses of Theorem \ref{thm:two-control} with $m=1$ is necessarily GSC.
\begin{lemma}[The case $m=1$]
    Under the hypotheses of Theorem \ref{thm:two-control} with \(m=1\), the unitary \(W\) is generalised semi-Clifford.
    \label{lem:the case m=1}
\end{lemma}
\begin{proof}
    Normalise $W$ so that \(C_{\boldsymbol 0}=\mathbb I_{\mathcal H_n}\). For $m=1$, hypothesis ii) of Theorem \ref{thm:two-control} becomes \(\mathsf R_{\boldsymbol x}=\mathbb I_{V_n}\), so $\sD_{\vx}(1,0)=\sD_{\vx}(0,0)$ for all $\vx\in \Ztwo^2$. Combining this with the independence of $\mathsf D_{\boldsymbol x}$ on the second control then gives $\mathsf D_{\boldsymbol x}(\boldsymbol z) = \mathsf D_{\boldsymbol x}(\boldsymbol 0)$, for all \(\boldsymbol x,\boldsymbol z\in\mathbb Z_2^2\). Hence
$$ \mathsf C_{\boldsymbol z+\boldsymbol x} \mathsf C_{\boldsymbol z}^{-1} = \mathsf D_{\boldsymbol x}(\boldsymbol z) = \mathsf D_{\boldsymbol x}(\boldsymbol 0) = \mathsf C_{\boldsymbol x}. $$
Taking $\boldsymbol x=\boldsymbol z$ in the above equation yields
$$
    \mathsf C_{\boldsymbol{z}}^2=\Id{V_n},
$$
while interchanging $\boldsymbol x$ and $\boldsymbol z$ gives
\begin{equation*}
    \mathsf C_{\boldsymbol z}\mathsf C_{\boldsymbol x}= \mathsf C_{\boldsymbol{x}}\mathsf C_{\boldsymbol z}.
\end{equation*}
Thus the matrices \(\mathsf C_{\boldsymbol z}\) commute and square to identity. Beigi-Shor showed that any family of commuting symplectic involutions preserves a common Lagrangian \cite[Theorem 4.2]{BS}. Theorem~\ref{thm:common-lagrangian} thus proves that \(W\) is GSC, and undoing normalisation preserves this conclusion.
\end{proof} 
Thus we will construct a gate at the fifth level of the hierarchy that is not GSC by applying Theorem \ref{thm:two-control} with $m=2$, in which case the required condition is $(\sR_{\vx}-\Id{V_n})^3=0$. Note that Lemma \ref{lem:the case m=1} does not exclude the possibility that a gate satisfying the hypotheses of Theorem \ref{thm:two-control} with $m>1$ may be in $\CC_4(n+2)$. We will show separately that the gate we construct is in fact in $\CC_5(n+2)\setminus\CC_4(n+2)$. 

\begin{status}{Construction requirements after Section 3}
Find real $(C_{\vz})_{\vz\in\Ztwo^2}$ in $\Cliff_n$, with $C_{\bm0}=\Id{\HH_n}$. For
\[
\sD_{\vx}(\vz)=\sC_{\vz+\vx}\sC_{\vz}^{-1},\qquad
\sR_{\vx}=\sD_{\vx}(1,0)\sD_{\vx}(0,0)^{-1},
\]
require, for all $\vx\in\Ztwo^2$ and $(a,b)\in\Ztwo^2$,
\[
\sD_{\vx}(a,b)=\sD_{\vx}(a,0),\qquad
(\sR_{\vx}-\Id{V_n})^3=0.
\]
These conditions give $W\in\CC_5(n+2)$. The same family must still have no common invariant Lagrangian for its matrices $\sC_{\vz}$.
\end{status}

\section{Reduction to a Clifford and a Pauli}
\label{sec:clifford-pauli}
By considering the comparison-independence equations in hypothesis (i) of Theorem \ref{thm:two-control}, we reduce choosing the family $(C_{\vz})_{\vz\in\Ztwo^2}$ to choosing two Clifford gates. Hypothesis (ii) of Theorem \ref{thm:two-control} imposes a nilpotence condition on each of the four relative matrices $\sR_{\vx}, \vx\in \Ztwo^2$. We see that one relative matrix is necessarily identity and force two of the others to be conjugate by a Clifford $J$, reducing the number of nilpotence conditions to check to two. Taking $J$ to be a Pauli rotation then reduces the hierarchy construction to a Clifford $B$ and a Hermitian Pauli $P$.

\subsection{The independence equations}
The solution of the comparison-independence condition imposes simple constraints on the symplectic matrices corresponding to the target Clifford gates.

\begin{lemma}[Solution of the independence equations]\label{lem:independence}
Let $G$ be a group with identity $1_G$, and let $f:\Ztwo^2\to G$ satisfy $f(\bm0)=1_G$. The coordinates of $\vz=(z_{\mathtt{c}},z_{\mathtt{d}})$ are ordered by the controls $\mathtt{c},\mathtt{d}$. The following conditions are equivalent.
\begin{enumerate}
\item For all $\vx\in\Ztwo^2$ and $\vz=(z_{\mathtt{c}},z_{\mathtt{d}})\in\Ztwo^2$,
\[
f(\vz+\vx)f(\vz)^{-1}
=f((z_{\mathtt{c}},0)+\vx)f((z_{\mathtt{c}},0))^{-1}.
\]
\item There exist $a,b\in G$, with $a^2=1_G$, such that
\begin{equation}
\text{for all }\vz\in\Ztwo^2,\qquad f(\vz)=b^{z_{\mathtt{c}}}a^{z_{\mathtt{d}}}.
\label{eq:group-pattern}
\end{equation}
\end{enumerate}
\end{lemma}
\begin{proof}
Assume (i), and put $a=f(0,1)$ and $b=f(1,0)$. Comparing $z_{\mathtt{d}}=0,1$ with $z_{\mathtt{c}}=0$ and $\vx=(0,1)$ gives $a=a^{-1}$. Doing the same with $\vx=(1,0)$ gives $b=f(1,1)a^{-1}$, hence $f(1,1)=ba$. This determines all four values as in (ii). Conversely, if (ii) holds, then
\[
f(\vz+\vx)f(\vz)^{-1}
=b^{z_{\mathtt{c}}+x_{\mathtt{c}}}a^{x_{\mathtt{d}}}b^{-z_{\mathtt{c}}},
\]
where $z_{\mathtt{c}}+x_{\mathtt{c}}$ is evaluated in $\Ztwo$ before taking the power. This is independent of $z_{\mathtt{d}}$.
\end{proof}

Applied to $G=\Sp(V_n,\omega)$, the lemma specifies the normalised symplectic family. We choose exact representatives $C_{\vz}=B^{z_{\mathtt{c}}}A^{z_{\mathtt{d}}}$, with $A,B\in\Cliff_n$, and write
\begin{equation}
U:=\ctrl_{\mathtt{c}}(B)\ctrl_{\mathtt{d}}(A)
=\sum_{\vz\in\Ztwo^2}\proj{\vz}\otimes B^{z_{\mathtt{c}}}A^{z_{\mathtt{d}}}\in\mathrm U(\HH_{n+2}).
\label{eq:product-U}
\end{equation}
This allows us to directly compute the relative matrices in terms of $\sA$ and $\sB$.
\begin{lemma}[The three nonzero displacements]\label{lem:comparison-table}
Let $A,B\in\Cliff_n$ satisfy $\sA^2=\Id{V_n}$, and set $C_{\vz}=B^{z_{\mathtt{c}}}A^{z_{\mathtt{d}}}$. For all $\vx\in\Ztwo^2$, the map $\sD_{\vx}(\vz)=\sC_{\vz+\vx}\sC_{\vz}^{-1}$ is independent of $z_{\mathtt{d}}$. The comparison and relative matrices are
\begin{center}
\renewcommand{\arraystretch}{1.28}
\begin{tabular}{@{}c c c c@{}}
\toprule
$\vx$ & $\sD_{\vx}(0,0)$ & $\sD_{\vx}(1,0)$ & $\sR_{\vx}$\\
\midrule
$(0,0)$ & $\Id{V_n}$ & $\Id{V_n}$ & $\Id{V_n}$\\
$(0,1)$ & $\sA$ & $\sB\sA\sB^{-1}$ & $\sB\sA\sB^{-1}\sA$\\
$(1,0)$ & $\sB$ & $\sB^{-1}$ & $\sB^{-2}$\\
$(1,1)$ & $\sB\sA$ & $(\sB\sA)^{-1}$ & $(\sB\sA)^{-2}$\\
\bottomrule
\end{tabular}
\end{center}
All entries in the last three columns lie in $\Sp(V_n,\omega)$.
\end{lemma}
\begin{proof}
    Lemma \ref{lem:independence} ensures that $\sD_{\vx}(\vz)$ is independent of $z_{\mathtt d}$ for all $\vx\in \Ztwo^2$. The table then follows by direct computation, using the fact that $\sA=\sA^{-1}$.
\end{proof}
The three nonzero displacements therefore leave three conditions $(\sR_{\vx}-\Id{V_n})^3=0$. Note that the hypothesis that the square is identity concerns the symplectic matrix $\sA$, not the unitary $A$.

\subsection{Conjugacy determines \texorpdfstring{$A$}{A}}
We now require $BA$ to be conjugate to $B$, which is not forced by the hierarchy test. We choose this because it makes the last two relative-matrix conditions equivalent, leaving only one nilpotence calculation for them. Solving this sufficient conjugation condition also removes $A$ as an independent choice.

Writing the conjugating Clifford as $BJB^\dagger$, with $J\in\Cliff_n$, the conjugation condition is equivalent to
\begin{equation}
A=JBJ^\dagger B^\dagger\in\Cliff_n.
\label{eq:commutator-A}
\end{equation}
Indeed,
\begin{equation}
BA=(BJB^\dagger)B(BJB^\dagger)^\dagger.
\label{eq:conjugate-B}
\end{equation}
After \eqref{eq:commutator-A}, it remains to ensure $\sA^2=\Id{V_n}$ and
\[
(\sB\sA\sB^{-1}\sA-\Id{V_n})^3=0,\qquad
(\sB^{-2}-\Id{V_n})^3=0.
\]
The next subsection specifies $J$ to meet these remaining requirements.

\subsection{Pauli rotations and the remaining conditions}
As we require $\sA$ to be an involution, $\sB\sA\sB^{-1}$ must also be an involution. The product of commuting involutions is an involution, so the first relative matrix $\sB\sA\sB^{-1}\sA$ is an involution if $\sA$ and $\sB\sA\sB^{-1}$ commute. In characteristic two,
\[
\sR^2=\Id{V_n}\quad\Longrightarrow\quad(\sR-\Id{V_n})^2=0,
\]
which is sufficient for the required cube to vanish. 
We therefore seek a choice of $J$ for which $\sA$ and $\sB\sA\sB^{-1}$ commute and square to identity. Pauli rotations provide a simple family: their corresponding symplectic matrices square to identity, their commutation is controlled by Pauli commutation, and Clifford conjugation replaces the Pauli by its conjugate. We thus take $J$ to be a Pauli rotation.

\begin{definition}[Pauli rotation]\label{def:rotation}
For a Hermitian $P\in\PP_n$, the Pauli rotation used below is
\begin{equation}
J_P:=\exp(-i\pi P/4)=\frac{\Id{\HH_n}-iP}{\sqrt2}\in\mathrm U(\HH_n).
\label{eq:rotation}
\end{equation}
\end{definition}

Clifford transvections were studied in e.g. \cite[Section~2.3, equations~(30)--(31)]{PRTC}.

\begin{lemma}[Symplectic matrix of the rotation]\label{lem:rotation}
Let $P\in\PP_n$ be Hermitian and put $\vp=[P]\in V_n$.
\begin{enumerate}
\item $J_P\in\Cliff_n$ and $J_P^2=-iP$.
\item For all $\vecv\in V_n$,
\begin{equation}
\sJ_P\vecv=\vecv+\omega(\vp,\vecv)\vp.
\label{eq:transvection}
\end{equation}
In particular, $\sJ_P^2=\Id{V_n}$.
\item For all Hermitian $Q\in\PP_n$ commuting with $P$, $J_PJ_Q=J_QJ_P$.
\end{enumerate}
\end{lemma}
For (ii), if $[Q]=\vecv$, expanding $J_PQJ_P^\dagger$ gives $Q$ when $PQ=QP$ and $-iPQ$ when $PQ=-QP$; passing to classes proves \eqref{eq:transvection}, whose square is identity since $\omega(\vp,\vp)=0$. The elementary identities in (i) and (iii) will also be used below.

\begin{lemma}[Three commuting Pauli conjugates]\label{lem:three-commuting}
Let $B\in\Cliff_n$ and let $P\in\PP_n$ be Hermitian. Set
\[
J=J_P,\qquad A=JBJ^\dagger B^\dagger\in\Cliff_n.
\]
Suppose $P$, $BPB^\dagger$, and $B^2PB^{-2}$ commute pairwise. Then
\begin{equation}
\sA^2=\Id{V_n},\qquad
\sA(\sB\sA\sB^{-1})=(\sB\sA\sB^{-1})\sA.
\label{eq:commuting-squares}
\end{equation}
Consequently,
\begin{equation}
(\sB\sA\sB^{-1}\sA-\Id{V_n})^2=0.
\label{eq:first-relative-involution}
\end{equation}
\end{lemma}
The calculations proving this lemma are given in Appendix~\ref{app:commuting-rotations}.

The remaining condition concerns $B$ alone. The notation $\sN=\sB-\Id{V_n}$ denotes an endomorphism of $V_n$; it is not assumed to be symplectic or invertible.

\begin{corollary}[Sufficient fifth-level construction]\label{cor:fifth-level}
Let $B\in\Cliff_n$ and let $P\in\PP_n$ be Hermitian. Define
\[
J:=J_P=\frac{\Id{\HH_n}-iP}{\sqrt2},\qquad
\sN:=\sB-\Id{V_n}\in\End_{\Ztwo}(V_n).
\]
Assume:
\begin{enumerate}
\item $B$ and $J$ are real;
\item $P$, $BPB^\dagger$, and $B^2PB^{-2}$ commute pairwise;
\item $\sN^6=0$.
\end{enumerate}
For
\[
A:=JBJ^\dagger B^\dagger\in\Cliff_n,\qquad
U:=\ctrl_{\mathtt{c}}(B)\ctrl_{\mathtt{d}}(A)\in\mathrm U(\HH_{n+2}),
\]
one has $U\in\CC_5(n+2)$.
\end{corollary}
\begin{proof}
The conditional operators $B^{z_{\mathtt{c}}}A^{z_{\mathtt{d}}}$ are real Clifford gates. Lemma~\ref{lem:three-commuting} gives $\sA^2=\Id{V_n}$ and $(\sR_{(0,1)}-\Id{V_n})^2=0$, so the second-control independence and first cubic test follow from Lemma~\ref{lem:comparison-table}. The zero displacement gives $\sR_{\bm0}=\Id{V_n}$.

For the second nonzero displacement, characteristic two gives
\begin{equation}
\sB^{-2}-\Id{V_n}=\sB^{-2}\sN^2,\qquad
(\sB^{-2}-\Id{V_n})^3=\sB^{-6}\sN^6=0,
\label{eq:inverse-square}
\end{equation}
since $\sN$ commutes with $\sB$ and $\sB^{-1}$. By \eqref{eq:conjugate-B}, the final relative matrix $(\sB\sA)^{-2}$ is conjugate to $\sB^{-2}$, so its difference from identity also has cube zero. Theorem~\ref{thm:two-control} with $m=2$ proves the claim.
\end{proof}
Conjugacy and the Pauli rotation have reduced hierarchy membership to conditions on $B$ and $P$. These hypotheses alone do not assert that $U$ is non-GSC or excluded from the fourth level.

\begin{status}{Construction requirements after Section 4}
Find $B\in\Cliff_n$ and a Hermitian $P\in\PP_n$. With
\[
J=\frac{\Id{\HH_n}-iP}{\sqrt2},\qquad \sN=\sB-\Id{V_n},
\]
require $B=\overline B$, $J=\overline J$, pairwise commutation of
\[
P,\quad BPB^\dagger,\quad B^2PB^{-2},
\]
and $\sN^6=0$. Define $A=JBJ^\dagger B^\dagger$ and $U=\ctrl_{\mathtt{c}}(B)\ctrl_{\mathtt{d}}(A)$. Then $U\in\CC_5(n+2)$. The same choice must still ensure that $\sA,\sB$ have no common invariant Lagrangian.
\end{status}

\section{Sufficient conditions to be non-GSC}
\label{sec:cyclic}
We strengthen the conditions on $B$ and $P$ to exclude a common invariant Lagrangian and membership in $\CC_4(5)$. A cyclic vector for $\sN=\sB-\Id{V_n}$ supplies the hierarchy bound and both obstructions; see also ~\cite[Sections~2.4 and~3]{XW}.

\subsection{The cyclic basis and the target dimension}
As the hierarchy condition already involves powers of $\sN$, we choose a basis made from successive powers of this same operator, so that applying $\sN$ simply moves to the next basis vector. Besides making the power calculations immediate, a full-length nilpotent cyclic basis makes the fixed space of $\sB$ one-dimensional. Its nonzero vector will then lie in every candidate invariant Lagrangian.

For $\sN\in\End_{\Ztwo}(V_n)$, a vector $\vp\in V_n$ is \emph{cyclic} precisely when
\[
\spanZ\{\sN^j\vp:j\in\ZZ_{\geq0}\}=V_n.
\]
We require its first $2n$ successive images to form a basis and its next image to vanish.

\begin{lemma}[Consequences of a cyclic basis]\label{lem:cyclic-basis}
Let $\sB\in\Sp(V_n,\omega)$ and $\vp\in V_n$. Set $\sN=\sB-\Id{V_n}$ and, for all $0\leq j\leq2n-1$, $\ve_j=\sN^j\vp$. Suppose $(\ve_0,\ldots,\ve_{2n-1})$ is a basis of $V_n$ and $\sN\ve_{2n-1}=\bm0$. Then:
\begin{enumerate}
\item $\sN^{2n}=0$ and $\sN^{2n-1}\vp=\ve_{2n-1}\ne\bm0$;
\item the kernel and image are
\begin{equation}
\ker\sN=\Ztwo\ve_{2n-1},\qquad
\im\sN=\spanZ\{\ve_1,\ldots,\ve_{2n-1}\};
\label{eq:cyclic-kernel}
\end{equation}
\item for all nonzero $\sB$-invariant subspaces $L\leq V_n$, $\ve_{2n-1}\in L$.
\end{enumerate}
\end{lemma}
\begin{proof}
Parts (i) and (ii) follow directly from $\sN\ve_{2n-1}=\bm0$ and, for all $0\leq j<2n-1$, $\sN\ve_j=\ve_{j+1}$. Since $\sB$ fixes $\ve_{2n-1}$ and the kernel in (ii) is one-dimensional, (iii) follows from Lemma~\ref{lem:nilpotent-invariance}(ii).
\end{proof}

\begin{lemma}[The three Pauli conjugates]\label{lem:three-classes}
Let $B\in\Cliff_n$ and let $P\in\PP_n$ be Hermitian. Define
\[
\vp=[P]\in V_n,\qquad \sN=\sB-\Id{V_n},\qquad
\text{for all }j\in\{0,1,2\},\quad\ve_j=\sN^j\vp.
\]
Then:
\begin{enumerate}
\item the three conjugate classes are
\begin{equation}
[P]=\ve_0,\qquad [BPB^\dagger]=\ve_0+\ve_1,\qquad
[B^2PB^{-2}]=\ve_0+\ve_2;
\label{eq:three-classes}
\end{equation}
\item $P$, $BPB^\dagger$, and $B^2PB^{-2}$ commute pairwise if and only if
\begin{equation}
\omega(\ve_0,\ve_1)=\omega(\ve_0,\ve_2)=0;
\label{eq:two-zero-pairings}
\end{equation}
\item if $\ve_0,\ve_1,\ve_2$ are linearly independent and \eqref{eq:two-zero-pairings} holds, then $n\geq3$.
\end{enumerate}
\end{lemma}
\begin{proof}
Since $\sB=\Id{V_n}+\sN$ and $\sB^2=\Id{V_n}+\sN^2$ in characteristic two, (i) follows. The three pairings of these classes are
\[
\omega(\vp,\sB\vp)=\omega(\ve_0,\ve_1),\qquad
\omega(\vp,\sB^2\vp)=\omega(\ve_0,\ve_2),
\]
\[
\omega(\sB\vp,\sB^2\vp)=\omega(\vp,\sB\vp)=\omega(\ve_0,\ve_1),
\]
where the last equality uses symplecticity. Equation~\eqref{eq:commutation} proves (ii). If $\ve_0,\ve_1,\ve_2$ are independent, so are $\ve_0,\ve_0+\ve_1,\ve_0+\ve_2$. Under \eqref{eq:two-zero-pairings} their span is isotropic of dimension three, which proves (iii).
\end{proof}
Imposing the construction requirement $\sN^6=0$ forces $n=3$.

\begin{lemma}[Selection of three targets]\label{lem:three-targets}
Let $n\geq2$ be an integer, let $\sB\in\Sp(V_n,\omega)$, and let $\vp\in V_n$. Put $\sN=\sB-\Id{V_n}$ and, for all $0\leq j\leq2n-1$, $\ve_j=\sN^j\vp$. Assume:
\begin{enumerate}
\item $(\ve_0,\ldots,\ve_{2n-1})$ is a basis of $V_n$ and $\sN\ve_{2n-1}=\bm0$;
\item $\omega(\ve_0,\ve_1)=\omega(\ve_0,\ve_2)=0$;
\item $\sN^6=0$.
\end{enumerate}
Then $n=3$.
\end{lemma}
\begin{proof}
The symplectic calculation in Lemma~\ref{lem:three-classes} makes $\vp,\sB\vp,\sB^2\vp$ three independent orthogonal vectors, so $n\geq3$. Lemma~\ref{lem:cyclic-basis} gives $\mathsf N^{2n-1}\boldsymbol p\ne0$; since
$\mathsf N^6=0$, this forces $2n-1<6$, hence $2n\le6$.
Therefore $n=3$.
\end{proof}
The cyclic-basis and commutation requirements select three targets \emph{within this construction}. This is not a lower bound for all possible counterexamples.

\subsection{A common invariant Lagrangian is impossible}
The last vector of the cyclic basis lies in every candidate invariant Lagrangian. We determine the action of the already prescribed $A$ on that vector and then use the resulting dimension obstruction.

\begin{lemma}[Pairings with the fixed vector]\label{lem:fixed-pairings}
Let $\sB\in\Sp(V_n,\omega)$ and $\vp\in V_n$. Set $\sN=\sB-\Id{V_n}$ and, for all $0\leq j\leq2n-1$, $\ve_j=\sN^j\vp$. Suppose $(\ve_0,\ldots,\ve_{2n-1})$ is a basis and $\sN\ve_{2n-1}=\bm0$. For $\vw:=\ve_{2n-1}$,
\begin{enumerate}
\item $\sB\vw=\vw$;
\item for all $1\leq j\leq2n-1$, $\omega(\vw,\ve_j)=0$;
\item $\omega(\vw,\vp)=\omega(\vw,\sB\vp)=1$.
\end{enumerate}
\end{lemma}
\begin{proof}
The equation $\sN\vw=\bm0$ gives $\sB\vw=\vw$, hence $\sB^{-1}\vw=\vw$. For all $\vecv\in V_n$,
\[
\omega(\sN\vecv,\vw)
=\omega(\sB\vecv,\vw)-\omega(\vecv,\vw)
=\omega(\vecv,\sB^{-1}\vw)-\omega(\vecv,\vw)=0.
\]
Every $\ve_j$ with $j\geq1$ lies in $\im\sN$, proving (ii). If $\omega(\vw,\vp)$ were also zero, $\vw$ would be orthogonal to the whole basis, contradicting nondegeneracy and $\vw\ne\bm0$. Thus $\omega(\vw,\vp)=1$, and
\[
\omega(\vw,\sB\vp)=\omega(\sB^{-1}\vw,\vp)=1.
\]
\end{proof}

\begin{lemma}[Action of the Pauli-rotation commutator]\label{lem:moves-fixed-vector}
Let $n\geq2$ be an integer, let $B\in\Cliff_n$, and let $P\in\PP_n$ be Hermitian. Define
\[
\vp=[P],\qquad \sN=\sB-\Id{V_n},\qquad
J=J_P,\qquad A=JBJ^\dagger B^\dagger.
\]
For all $0\leq j\leq2n-1$, set $\ve_j=\sN^j\vp$. Assume:
\begin{enumerate}
\item $(\ve_0,\ldots,\ve_{2n-1})$ is a basis of $V_n$;
\item $\sN\ve_{2n-1}=\bm0$;
\item $\omega(\ve_0,\ve_1)=0$.
\end{enumerate}
Then
\begin{equation}
(\sA-\Id{V_n})\ve_{2n-1}=\ve_1.
\label{eq:moves-fixed-vector}
\end{equation}
\end{lemma}
\begin{proof}
Put $\vw=\ve_{2n-1}$. Lemma~\ref{lem:fixed-pairings} gives $\sB^{-1}\vw=\vw$ and $\omega(\vp,\vw)=1$, while
\[
\omega(\vp,\sB\vp)=\omega(\ve_0,\ve_0+\ve_1)=0.
\]
By Lemma~\ref{lem:rotation}, $\sJ^{-1}=\sJ$ and $\sJ\vecv=\vecv+\omega(\vp,\vecv)\vp$. Applying the four factors in the definition of $\sA$ directly, we obtain
\begin{align*}
\sA\vw
&=\sJ\sB\sJ^{-1}\sB^{-1}\vw
=\sJ\sB(\vw+\vp)\\
&=\sJ(\vw+\sB\vp)
=\vw+\sB\vp+\bigl(\omega(\vp,\vw)+\omega(\vp,\sB\vp)\bigr)\vp\\
&=\vw+\sB\vp+\vp=\vw+\ve_1.\qedhere
\end{align*}
\end{proof}

\begin{lemma}[Invariant-subspace obstruction]\label{lem:subspace-obstruction}
Let $n\geq2$ be an integer, let $\sB\in\Sp(V_n,\omega)$, and let $\vp\in V_n$. Set $\sN=\sB-\Id{V_n}$ and, for all $0\leq j\leq2n-1$, $\ve_j=\sN^j\vp$. Suppose $(\ve_0,\ldots,\ve_{2n-1})$ is a basis and $\sN\ve_{2n-1}=\bm0$. Let $\sA\in\End_{\Ztwo}(V_n)$ satisfy
\[
(\sA-\Id{V_n})\ve_{2n-1}=\ve_1.
\]
For all nonzero subspaces $L\leq V_n$ invariant under both $\sA$ and $\sB$,
\begin{equation}
\spanZ\{\ve_1,\ldots,\ve_{2n-1}\}\subseteq L.
\label{eq:too-large-subspace}
\end{equation}
In particular, $\sA$ and $\sB$ have no common invariant Lagrangian.
\end{lemma}
\begin{proof}
By Lemma~\ref{lem:cyclic-basis}(iii), $\ve_{2n-1}\in L$. Invariance under $\sA$ gives
\[
\ve_1=(\sA-\Id{V_n})\ve_{2n-1}\in L.
\]
Since $L$ is also invariant under $\sN=\sB-\Id{V_n}$, for all $1\leq j\leq2n-1$,
\[
\ve_j=\sN^{j-1}\ve_1\in L.
\]
These vectors are independent, proving \eqref{eq:too-large-subspace}. A Lagrangian has dimension $n$, whereas $2n-1>n$ for $n\geq2$, so it cannot satisfy this inclusion.
\end{proof}
The operator $A$ already chosen for hierarchy membership therefore satisfies the obstruction.

\subsection{An exact fourth-level obstruction}
One Pauli commutator gives a necessary condition for membership in $\CC_4$. 

\begin{lemma}[Commutator with the first control]\label{lem:fourth-obstruction}
Let $A,B\in\Cliff_n$. Define
\[
U=\ctrl_{\mathtt{c}}(B)\ctrl_{\mathtt{d}}(A)\in\mathrm U(\HH_{n+2}),\qquad
K=\Id{\HH_{\mathtt{d}}}\otimes B^{-2}\in\Cliff_{n+1}.
\]
Here $X_{\mathtt{c}}$ is Pauli $X$ on the first control and identity elsewhere.
\begin{enumerate}
\item In the register order $\mathtt{c},\mathtt{d},\text{targets}$,
\begin{equation}
\Delta_{X_{\mathtt{c}}}U=\ctrl_{\mathtt{c}}(K)(\Id{\HH_{\mathtt{c}}}\otimes\Id{\HH_{\mathtt{d}}}\otimes B).
\label{eq:first-control-commutator}
\end{equation}
\item If $U\in\CC_4(n+2)$, then
\[
\Id{\HH_{\mathtt{d}}}\otimes B^{-4}\in\PP_{n+1},\qquad \sB^4=\Id{V_n}.
\]
\item If $(\sB-\Id{V_n})^4\ne0$, then $U\notin\CC_4(n+2)$.
\end{enumerate}
\end{lemma}
\begin{proof}
For all $(y,z_{\mathtt{d}})\in\Ztwo^2$, the comparison operator for displacement $(1,0)$ and target Pauli $\Id{\mathcal{H}_n}$ is
\[
B^{1-y}A^{z_{\mathtt{d}}}(B^yA^{z_{\mathtt{d}}})^\dagger=B^{1-2y},
\]
where the exponents in this equality are ordinary integers. Thus Lemma~\ref{lem:all-commutators} gives
\begin{align*}
\Delta_{X_{\mathtt{c}}}U
&=\sum_{y\in\Ztwo}\proj y_{\mathtt{c}}\otimes\Id{\HH_{\mathtt{d}}}\otimes B^{1-2y}\\
&=\ctrl_{\mathtt{c}}(\Id{\HH_{\mathtt{d}}}\otimes B^{-2})(\Id{\HH_{\mathtt{c}}}\otimes\Id{\HH_{\mathtt{d}}}\otimes B),
\end{align*}
proving (i).

If $U\in\CC_4(n+2)$, its Pauli commutator lies in $\CC_3(n+2)$. Removing the right Clifford factor in (i), then applying Lemma~\ref{lem:controlled-powers} with $m=1$, gives
\[
K^2=\Id{\HH_{\mathtt{d}}}\otimes B^{-4}\in\PP_{n+1}.
\]
Its corresponding symplectic matrix is identity. On target Pauli classes this matrix is $\sB^{-4}$, so $\sB^4=\Id{V_n}$. Finally,
\[
(\sB-\Id{V_n})^4=\sB^4-\Id{V_n}
\]
in characteristic two, proving (iii).
\end{proof}
A six-vector cyclic basis for $\sN$ has $\sN^4\ve_0=\ve_4\ne\bm0$, so it supplies this obstruction. Thus the gate we will construct is in $\CC_5(5)\setminus\CC_4(5)$.

\subsection{The abstract construction}
We collect the conditions that give all three conclusions. The remaining sections construct exact operators satisfying these conditions.

Fifth-level membership is Corollary~\ref{cor:fifth-level}, using Lemmas~\ref{lem:cyclic-basis} and~\ref{lem:three-classes}; failure of GSC follows from Lemmas~\ref{lem:moves-fixed-vector}--\ref{lem:subspace-obstruction} and Theorem~\ref{thm:common-lagrangian}; exclusion from the fourth level is Lemma~\ref{lem:fourth-obstruction}. All three conclusions follow from the same six-vector configuration. The only remaining task is to realise these requirements by a real Clifford $B$ and a Hermitian Pauli $P$.

\begin{status}{Construction requirements after Section 5}
Find $B\in\Cliff_3$ and a Hermitian $P\in\PP_3$. Set
\[
J=\frac{\Id{\HH_3}-iP}{\sqrt2}\in\Cliff_3,\qquad
\sN=\sB-\Id{V_3}\in\End_{\Ztwo}(V_3),
\]
and, for all $0\leq j\leq5$, $\ve_j=\sN^j[P]\in V_3$.
Require $B=\overline B$, $J=\overline J$, and
\[
(\ve_0,\ldots,\ve_5)\text{ is a basis of }V_3,\qquad
\sN\ve_5=\bm0,\qquad
\omega(\ve_0,\ve_1)=\omega(\ve_0,\ve_2)=0.
\]
With $A=JBJ^\dagger B^\dagger\in\Cliff_3$, the resulting $U=\ctrl_{\mathtt{c}}(B)\ctrl_{\mathtt{d}}(A)\in\mathrm U(\HH_5)$ satisfies
\[
U\in\CC_5(5)\setminus\CC_4(5),\qquad U\text{ is not GSC}.
\]
\end{status}

\section{Realisation of the symplectic matrix}
\label{sec:symplectic-realisation}
We determine the pairings and quadratic values compatible with the prescribed action. A Pauli basis realising these data leaves only the construction of a real Clifford representative.

\subsection{The invariant alternating form}
We prescribe the linear action before choosing a basis. Its invariance equations and the two commutation conditions determine the required pairings.

We seek an ordered basis $(\ve_0,\ldots,\ve_5)$ of $V_3$ for which the linear map $\sB\in\GL(V_3)$ determined by
\begin{equation}
\sB\ve_5=\ve_5,\qquad
\text{for all }0\leq j<5,\quad \sB\ve_j=\ve_j+\ve_{j+1},
\label{eq:prescribed-action}
\end{equation}
preserves the standard symplectic form. For any such basis, $\sN=\sB-\Id{V_3}$ has $\ve_0$ as a cyclic vector, with $\sN^6=0$. In particular, $\sB$ is invertible, with inverse $\Id{V_3}+\sN+\cdots+\sN^5$.

\begin{lemma}[Equations for an invariant alternating form]\label{lem:form-equations}
Let $(\ve_0,\ldots,\ve_5)$ be a basis of $V_3$, and let $\sB\in\GL(V_3)$ be given by \eqref{eq:prescribed-action}. For all $0\leq i,j\leq5$, put $g_{ij}:=\omega(\ve_i,\ve_j)\in\Ztwo$; for all $0\leq i\leq6$, set $g_{i6}=g_{6i}=0$.
The form is $\sB$-invariant if and only if, for all $0\leq i,j\leq5$,
\begin{equation}
g_{i,j+1}+g_{i+1,j}+g_{i+1,j+1}=0.
\label{eq:form-equations}
\end{equation}
\end{lemma}
\begin{proof}
Set $\ve_6=\bm0$. For all pairs of basis vectors,
\begin{align*}
\omega(\sB\ve_i,\sB\ve_j)
&=\omega(\ve_i+\ve_{i+1},\ve_j+\ve_{j+1})\\
&=g_{ij}+g_{i,j+1}+g_{i+1,j}+g_{i+1,j+1}.
\end{align*}
Equality with $g_{ij}$ is exactly \eqref{eq:form-equations}. Bilinearity makes equality on all basis pairs equivalent to invariance on $V_3\times V_3$.
\end{proof}

\begin{lemma}[The required pairings]\label{lem:required-form}
Let $(\ve_0,\ldots,\ve_5)$ be a basis of $V_3$, and let $\sB\in\GL(V_3)$ be given by \eqref{eq:prescribed-action}. Then $\sB\in\Sp(V_3,\omega)$ and
\[
\omega(\ve_0,\ve_1)=\omega(\ve_0,\ve_2)=0
\]
if and only if the matrix of pairings is
\begin{equation}
\bigl(\omega(\ve_i,\ve_j)\bigr)_{i,j=0}^{5}=\Gstar:=
\begin{pmatrix}
0&0&0&1&0&1\\
0&0&0&1&1&0\\
0&0&0&1&0&0\\
1&1&1&0&0&0\\
0&1&0&0&0&0\\
1&0&0&0&0&0
\end{pmatrix}\in\Ztwo^{6\times6}.
\label{eq:Gstar}
\end{equation}
\end{lemma}
\begin{proof}
The equations in Lemma~\ref{lem:form-equations}, together with $g_{ii}=0$, $g_{ij}=g_{ji}$, and $g_{01}=g_{02}=0$, have \eqref{eq:Gstar} as their unique nonzero solution. This matrix is nonsingular. The elimination and the nonsingularity check are given in Appendix~\ref{app:form-calculations}.
\end{proof}
The required matrix of pairings is therefore $\Gstar$. Its Pauli realisation must also respect the reality of $B$ and of the rotation about $P$.

\subsection{Quadratic values and reality}
The symplectic matrices corresponding to the real Clifford group preserve the standard quadratic form, in addition to the symplectic form~\cite{NRS}. The requirement that the Clifford $B$ we construct be real thus forces $\sB$ to preserve the standard quadratic form as well. We record the needed consequences.

In the Pauli coordinates $V_n\cong\Ztwo^n\oplus\Ztwo^n$, define
\begin{equation}
q:V_n\longrightarrow\Ztwo,\qquad q(\vu,\vecv)=\vu\cdot\vecv.
\label{eq:quadratic-form}
\end{equation}
For all $\va,\vb\in V_n$,
\begin{equation}
q(\va+\vb)=q(\va)+q(\vb)+\omega(\va,\vb).
\label{eq:polarisation}
\end{equation}

\begin{lemma}[Reality]\label{lem:reality}
Let $q:V_n\to\Ztwo$ be given by \eqref{eq:quadratic-form}, let $P\in\PP_n$ be Hermitian, and let $\vp=[P]\in V_n$. Then:
\begin{enumerate}
\item $\overline P=(-1)^{q(\vp)}P$;
\item $J_P$ is real if and only if $q(\vp)=1$;
\item if $B\in\Cliff_n$ is real, then, for all $\vecv\in V_n$,
\begin{equation}
q(\sB\vecv)=q(\vecv).
\label{eq:q-invariance}
\end{equation}
\end{enumerate}
\end{lemma}
\begin{proof}
The value $q(\vp)$ is the parity of the number of $Y$ factors in a Hermitian tensor-product Pauli representative. Only $Y$ changes sign under complex conjugation, proving (i). Thus $J_P=(\Id{\HH_n}-iP)/\sqrt2$ is real exactly when $P$ is purely imaginary, proving (ii).

For (iii), choose a Hermitian $Q\in\PP_n$ with $[Q]=\vecv$. Since $B$ is real,
\[
\overline{BQB^\dagger}
=B\overline Q B^\dagger
=(-1)^{q(\vecv)}BQB^\dagger.
\]
The Hermitian Pauli $BQB^\dagger$ has class $\sB\vecv$, so (i) also gives
\[
\overline{BQB^\dagger}=(-1)^{q(\sB\vecv)}BQB^\dagger.
\]
As $BQB^\dagger\ne0$, the two signs agree, proving \eqref{eq:q-invariance}.
\end{proof}

\begin{lemma}[Quadratic values in the cyclic basis]\label{lem:quadratic-values}
Let $(\ve_0,\ldots,\ve_5)$ be a basis of $V_3$ whose matrix of pairings is $\Gstar$ in \eqref{eq:Gstar}, and let $\sB\in\GL(V_3)$ be given by \eqref{eq:prescribed-action}. Let $q:V_3\to\Ztwo$ be given by \eqref{eq:quadratic-form}.
Then $q\circ\sB=q$ and $q(\ve_0)=1$ if and only if
\begin{equation}
\bigl(q(\ve_0),q(\ve_1),q(\ve_2),q(\ve_3),q(\ve_4),q(\ve_5)\bigr)
=(1,0,0,1,0,0).
\label{eq:quadratic-values}
\end{equation}
\end{lemma}
\begin{proof}
For all $0\leq j<5$, preservation on the basis vector $\ve_j$ is equivalent to
\begin{align*}
q(\ve_j)
&=q(\ve_j+\ve_{j+1})\\
&=q(\ve_j)+q(\ve_{j+1})+g_{j,j+1},
\end{align*}
that is, $q(\ve_{j+1})=g_{j,j+1}$. The adjacent entries of $\Gstar$ are $0,0,1,0,0$; the fixed vector $\ve_5$ gives no further condition. Together with $q(\ve_0)=1$, this proves necessity of \eqref{eq:quadratic-values}.

Conversely, these values make $q\circ\sB=q$ on the basis. The function $d:V_3\to\Ztwo$, $d(\vecv):=q(\sB\vecv)+q(\vecv)$, is linear, because polarisation and symplectic invariance give
\[
d(\va+\vb)=d(\va)+d(\vb)+\omega(\sB\va,\sB\vb)+\omega(\va,\vb)
=d(\va)+d(\vb).
\]
It vanishes on a basis and therefore on all of $V_3$.
\end{proof}
The requirement that $q(\ve_0)=1$ ensures that the rotation $J$ we construct is also real, by Lemma \ref{lem:reality}. These are necessary reality conditions for the desired Clifford representative. Its existence will be established by an explicit circuit.

\subsection{A Pauli basis with the prescribed data}
We realise $\Gstar$ and the six quadratic values in $V_3$. The coordinate choice is not unique; the pairings and the action have already been prescribed.

The pair $\ve_0,\ve_5$ must have pairing $1$ and quadratic values $1,0$. A simple choice is
\[
\ve_0=[Y_{\mathtt{3}}],\qquad \ve_5=[X_{\mathtt{3}}].
\]
The pair $\ve_1,\ve_4$ must have pairing $1$, both quadratic values zero, and zero pairings with that first pair. Choose
\[
\ve_1=[Z_{\mathtt{2}}],\qquad \ve_4=[X_{\mathtt{2}}].
\]
We may then take $\ve_2=[Z_{\mathtt{1}}]$. These choices are made for sparsity and simplicity; all required pairings among the five chosen vectors already agree with $\Gstar$.

The final vector is forced by the remaining equations. Write
\[
\ve_3=(u_1,u_2,u_3;v_1,v_2,v_3)\in\Ztwo^3\oplus\Ztwo^3.
\]
Its pairings with the five chosen vectors and its quadratic value require
\begin{equation}
\begin{aligned}
u_3+v_3&=1, & u_2&=1, & u_1&=1,\\
v_2&=0, & v_3&=0, & u_1v_1+u_2v_2+u_3v_3&=1.
\end{aligned}
\label{eq:pauli-coordinates}
\end{equation}
The unique solution is $(u_1,u_2,u_3;v_1,v_2,v_3)=(1,1,1;1,0,0)$, giving $\ve_3=[Y_{\mathtt{1}}X_{\mathtt{2}}X_{\mathtt{3}}]$.

\begin{lemma}[Pauli realisation]\label{lem:pauli-realisation}
In $V_3$, define
\begin{equation}
\begin{aligned}
\ve_0&=[Y_{\mathtt{3}}], & \ve_3&=[Y_{\mathtt{1}}X_{\mathtt{2}}X_{\mathtt{3}}],\\
\ve_1&=[Z_{\mathtt{2}}], & \ve_4&=[X_{\mathtt{2}}],\\
\ve_2&=[Z_{\mathtt{1}}], & \ve_5&=[X_{\mathtt{3}}].
\end{aligned}
\label{eq:pauli-basis}
\end{equation}
Then:
\begin{enumerate}
\item $(\ve_0,\ldots,\ve_5)$ is a basis of $V_3$;
\item $\bigl(\omega(\ve_i,\ve_j)\bigr)_{i,j=0}^5=\Gstar$;
\item for $q:V_3\to\Ztwo$ given by $q(\vu,\vecv)=\vu\cdot\vecv$,
\[
\bigl(q(\ve_0),\ldots,q(\ve_5)\bigr)=(1,0,0,1,0,0).
\]
\end{enumerate}
\end{lemma}
\begin{proof}
The choices above and \eqref{eq:pauli-coordinates} give (ii) and (iii). The matrix $\Gstar$ is nonsingular by Lemma~\ref{lem:required-form}, so the six vectors are independent and hence form a basis of $V_3$.
\end{proof}

\begin{corollary}[The symplectic matrix in the realised basis]\label{cor:realised-matrix}
Let $(\ve_0,\ldots,\ve_5)$ be the basis of $V_3$ in Lemma~\ref{lem:pauli-realisation}, and define $\sB:V_3\to V_3$ by
\begin{equation}
\sB\ve_5=\ve_5,\qquad\text{for all }0\leq j<5,\quad\sB\ve_j=\ve_j+\ve_{j+1}.
\label{eq:realised-action}
\end{equation}
Then:
\begin{enumerate}
\item $\sB\in\Sp(V_3,\omega)$ and $q\circ\sB=q$;
\item for $\sN=\sB-\Id{V_3}$, $\sN\ve_5=\bm0$ and, for all $0\leq j\leq5$, $\sN^j\ve_0=\ve_j$.
\end{enumerate}
\end{corollary}
\begin{proof}
Part (i) follows from Lemmas~\ref{lem:required-form}, \ref{lem:quadratic-values}, and~\ref{lem:pauli-realisation}; part (ii) is the prescribed action minus identity.
\end{proof}
Taking $P=Y_{\mathtt{3}}$ fixes the Hermitian representative of $\vp=\ve_0$ and its rotation $J=J_P$. Every real Clifford with corresponding symplectic matrix $\sB$ specified by \eqref{eq:realised-action} satisfies the construction requirements after Section~\ref{sec:cyclic} with this $P$.

\begin{status}{Construction requirements after Section 6}
Let $P=Y_{\mathtt{3}}\in\PP_3$, and let
\[
(\ve_0,\ldots,\ve_5)=([Y_{\mathtt{3}}],[Z_{\mathtt{2}}],[Z_{\mathtt{1}}],[Y_{\mathtt{1}}X_{\mathtt{2}}X_{\mathtt{3}}],[X_{\mathtt{2}}],[X_{\mathtt{3}}]).
\]
It remains to find $B\in\Cliff_3$ satisfying
\[
B=\overline B,\qquad\sB\ve_5=\ve_5,\qquad\text{for all }0\leq j<5,\quad\sB\ve_j=\ve_j+\ve_{j+1}.
\]
For all such $B$, the choices
\[
J=\frac{\Id{\HH_3}-iY_{\mathtt{3}}}{\sqrt2},\qquad A=JBJ^\dagger B^\dagger
\]
give $U=\ctrl_{\mathtt{c}}(B)\ctrl_{\mathtt{d}}(A)$ in $\CC_5(5)\setminus\CC_4(5)$ and not GSC.
\end{status}

\section{Exact Clifford realisation}
\label{sec:exact-realisation}
We construct a real Clifford with the prescribed corresponding symplectic matrix. We then simplify the Pauli-rotation commutator  and express the resulting five-qubit unitary using elementary controlled gates.

\subsection{A Clifford representative of \texorpdfstring{$\sB$}{B}}
The prescribed action first determines the images of the standard Pauli classes. An elementary Clifford circuit realises those images.

We use the standard Hadamard $H_{\mathtt{j}}$, controlled-$X$ operator $\CNOT_{\mathtt{j}\to\mathtt{k}}$, and controlled-$Z$ operator $\CZ_{\mathtt{j},\mathtt{k}}$ on the indicated qubits, with identity on unmentioned factors. In $\CNOT_{\mathtt{j}\to\mathtt{k}}$, $\mathtt{j}$ is the control.

\begin{lemma}[Action on standard Pauli classes]\label{lem:standard-classes}
In $V_3$, let
\[
(\ve_0,\ldots,\ve_5)=([Y_{\mathtt{3}}],[Z_{\mathtt{2}}],[Z_{\mathtt{1}}],[Y_{\mathtt{1}}X_{\mathtt{2}}X_{\mathtt{3}}],[X_{\mathtt{2}}],[X_{\mathtt{3}}]).
\]
For a linear map $\sB:V_3\to V_3$, the following conditions are equivalent.
\begin{enumerate}
\item $\sB\ve_5=\ve_5$ and, for all $0\leq j<5$, $\sB\ve_j=\ve_j+\ve_{j+1}$.
\item Its action on the standard Pauli classes is
\begin{equation}
\begin{aligned}
\sB[X_{\mathtt{1}}]&=[Z_{\mathtt{1}}], & \sB[Z_{\mathtt{1}}]&=[X_{\mathtt{1}}X_{\mathtt{2}}X_{\mathtt{3}}],\\
\sB[X_{\mathtt{2}}]&=[X_{\mathtt{2}}X_{\mathtt{3}}], & \sB[Z_{\mathtt{2}}]&=[Z_{\mathtt{1}}Z_{\mathtt{2}}],\\
\sB[X_{\mathtt{3}}]&=[X_{\mathtt{3}}], & \sB[Z_{\mathtt{3}}]&=[Z_{\mathtt{2}}Z_{\mathtt{3}}].
\end{aligned}
\label{eq:standard-classes}
\end{equation}
\end{enumerate}
\end{lemma}
\begin{proof}
By direct calculation in the two bases; see Appendix~\ref{app:standard-classes}.
\end{proof}
The exchange of $X_{\mathtt{1}}$ and $Z_{\mathtt{1}}$ suggests a Hadamard on qubit $\mathtt{1}$. After that exchange, the required $X$ images are obtained by adding qubit $\mathtt{2}$ to qubit $\mathtt{1}$, then qubit $\mathtt{3}$ to qubit $\mathtt{2}$. This gives $\CNOT_{\mathtt{1}\to\mathtt{2}}$ followed by $\CNOT_{\mathtt{2}\to\mathtt{3}}$.

\begin{lemma}[Exact real Clifford representative]\label{lem:exact-B}
On three qubits define
\begin{equation}
B:=\CNOT_{\mathtt{2}\to\mathtt{3}}\CNOT_{\mathtt{1}\to\mathtt{2}}H_{\mathtt{1}}\in\Cliff_3.
\label{eq:exact-B}
\end{equation}
Then:
\begin{enumerate}
\item $B=\overline B$;
\item its conjugation action is
\begin{equation}
\begin{aligned}
BX_{\mathtt{1}}B^\dagger&=Z_{\mathtt{1}}, & BZ_{\mathtt{1}}B^\dagger&=X_{\mathtt{1}}X_{\mathtt{2}}X_{\mathtt{3}},\\
BX_{\mathtt{2}}B^\dagger&=X_{\mathtt{2}}X_{\mathtt{3}}, & BZ_{\mathtt{2}}B^\dagger&=Z_{\mathtt{1}}Z_{\mathtt{2}},\\
BX_{\mathtt{3}}B^\dagger&=X_{\mathtt{3}}, & BZ_{\mathtt{3}}B^\dagger&=Z_{\mathtt{2}}Z_{\mathtt{3}};
\end{aligned}
\label{eq:exact-B-images}
\end{equation}
in particular, its corresponding symplectic matrix is the one specified in Lemma~\ref{lem:standard-classes};
\item for $P=Y_{\mathtt{3}}\in\PP_3$, the first three conjugates are
\begin{equation}
P=Y_{\mathtt{3}},\qquad BPB^\dagger=Z_{\mathtt{2}}Y_{\mathtt{3}},\qquad B^2PB^{-2}=Z_{\mathtt{1}}Y_{\mathtt{3}}.
\label{eq:exact-conjugates}
\end{equation}
\end{enumerate}
\end{lemma}
\begin{proof}
By direct calculation using the Hadamard and CNOT conjugation rules; all intermediate images and their signs are given in Appendix~\ref{app:exact-B}.
\end{proof}
The chosen $B$ is a real representative of the prescribed matrix. With $P=Y_{\mathtt{3}}$, it realises the cyclic basis and pairings in the boxed construction requirements after Section~\ref{sec:cyclic}.

\subsection{The exact rotation and commutator}
The Pauli $P$ determines its rotation $J$, and $A$ remains the commutator chosen in Section~\ref{sec:clifford-pauli}. 

For $P=Y_{\mathtt{3}}\in\PP_3$, direct multiplication of the one-qubit matrices gives
\begin{equation}
J=J_P=\frac{\Id{\HH_3}-iY_{\mathtt{3}}}{\sqrt2}=H_{\mathtt{3}}Z_{\mathtt{3}}.
\label{eq:exact-J}
\end{equation}
In particular, $J=\overline J$; the calculation is given in Appendix~\ref{app:exact-J}.

\begin{lemma}[Exact commutator circuit]\label{lem:exact-A}
On three qubits let
\[
B=\CNOT_{\mathtt{2}\to\mathtt{3}}\CNOT_{\mathtt{1}\to\mathtt{2}}H_{\mathtt{1}}\in\Cliff_3,\qquad J=H_{\mathtt{3}}Z_{\mathtt{3}}\in\Cliff_3.
\]
Define $A=JBJ^\dagger B^\dagger\in\Cliff_3$. Then:
\begin{enumerate}
\item $A=\CNOT_{\mathtt{2}\to\mathtt{3}}\CZ_{\mathtt{2},\mathtt{3}}$;
\item $A^2=Z_{\mathtt{2}}$ on $\HH_3$;
\item $\sA^2=\Id{V_3}$ on $V_3$.
\end{enumerate}
\end{lemma}
\begin{proof}
By direct calculation of $JBJ^\dagger B^\dagger$ and its square; see Appendix~\ref{app:exact-A}.
\end{proof}
The exact operators now satisfy the abstract construction. In particular, the requirement $\sA^2=\Id{V_3}$ applies to the symplectic matrix.

\subsection{The elementary controlled-gate expression}
For distinct qubits $\mathtt{a},\mathtt{b},\mathtt{j}$ of a five-qubit register, let $\Toffoli_{\mathtt{a},\mathtt{b}\to\mathtt{j}}$ and $\CCZ_{\mathtt{a},\mathtt{b},\mathtt{j}}$ denote the corresponding unitaries in $\mathrm U(\HH_5)$. For all $(s,t,u)\in\Ztwo^3$, their actions on the indicated qubits are
\begin{align}
\Toffoli_{\mathtt{a},\mathtt{b}\to\mathtt{j}}\ket{s,t,u}&=\ket{s,t,u+st},\label{eq:Toffoli}\\
\CCZ_{\mathtt{a},\mathtt{b},\mathtt{j}}\ket{s,t,u}&=(-1)^{stu}\ket{s,t,u}.
\label{eq:CCZ}
\end{align}
They act as identity on all other qubits.

\begin{lemma}[Elementary controlled-gate circuit]\label{lem:elementary-U}
Let $\mathtt{c},\mathtt{d}$ be control qubits and $\mathtt{1},\mathtt{2},\mathtt{3}$ target qubits. With
\[
A=\CNOT_{\mathtt{2}\to\mathtt{3}}\CZ_{\mathtt{2},\mathtt{3}},\qquad
B=\CNOT_{\mathtt{2}\to\mathtt{3}}\CNOT_{\mathtt{1}\to\mathtt{2}}H_{\mathtt{1}},
\]
the unitary $U=\ctrl_{\mathtt{c}}(B)\ctrl_{\mathtt{d}}(A)$ may be expressed as
\begin{equation}
\begin{aligned}
U={}&\Toffoli_{\mathtt{c},\mathtt{2}\to\mathtt{3}}\,\Toffoli_{\mathtt{c},\mathtt{1}\to\mathtt{2}}\,\ctrl_{\mathtt{c}}(H_{\mathtt{1}})\Toffoli_{\mathtt{d},\mathtt{2}\to\mathtt{3}}\,\CCZ_{\mathtt{d},\mathtt{2},\mathtt{3}}.
\end{aligned}
\label{eq:elementary-U}
\end{equation}
\end{lemma}
\begin{proof}
This follows from the controlled-product identity $\ctrl_{\mathtt{a}}(EF)=\ctrl_{\mathtt{a}}(E)\ctrl_{\mathtt{a}}(F)$. Since $\ctrl_{\mathtt{a}}(\CNOT_{\mathtt{b}\to\mathtt{j}})=\Toffoli_{\mathtt{a},\mathtt{b}\to\mathtt{j}}$ and $\ctrl_{\mathtt{a}}(\CZ_{\mathtt{b},\mathtt{j}})=\CCZ_{\mathtt{a},\mathtt{b},\mathtt{j}}$ by \eqref{eq:Toffoli}--\eqref{eq:CCZ}, controlling the products for $A$ and $B$ gives \eqref{eq:elementary-U}.
\end{proof}

\section{The five-qubit counterexample and its inverse}
\label{sec:counterexample}
The preceding construction yields the following unitary.

\begin{theorem}[Five-qubit counterexample]\label{thm:counterexample}
Let $\mathtt{c},\mathtt{d}$ be control qubits and let $\mathtt{1},\mathtt{2},\mathtt{3}$ be target qubits. Define the target Clifford gates
\[
A:=\CNOT_{\mathtt{2}\to\mathtt{3}}\CZ_{\mathtt{2},\mathtt{3}}\in\Cliff_3,\qquad
B:=\CNOT_{\mathtt{2}\to\mathtt{3}}\CNOT_{\mathtt{1}\to\mathtt{2}}H_{\mathtt{1}}\in\Cliff_3,
\]
and set
\begin{equation}
U:=\ctrl_{\mathtt{c}}(B)\ctrl_{\mathtt{d}}(A)
=\sum_{\vz\in\Ztwo^2}\proj{\vz}_{\mathtt{c},\mathtt{d}}\otimes B^{z_{\mathtt{c}}}A^{z_{\mathtt{d}}}\in\mathrm U(\HH_5).
\label{eq:final-U}
\end{equation}
Then:
\begin{enumerate}
\item $U\in\CC_5(5) \setminus \CC_4(5)$;
\item $U$ is not generalised semi-Clifford.
\end{enumerate}
\end{theorem}
\begin{proof}
Put $P=Y_3$, $\boldsymbol p=[P]=\boldsymbol e_0$,
$J=(\mathbb I_{\mathcal H_3}-iP)/\sqrt2=H_3Z_3$, and
$\mathsf N=\mathsf B-\mathbb I_{V_3}$. By Lemma~6.5,
Corollary~6.6, and Lemmas~7.1--7.2, $B$ is real and the displayed vectors
$\boldsymbol e_j=\mathsf N^j\boldsymbol p$ form a basis, with
$\mathsf N\boldsymbol e_5=0$ and
$\omega(\boldsymbol e_0,\boldsymbol e_1)
=\omega(\boldsymbol e_0,\boldsymbol e_2)=0$. Thus
$\mathsf N^6=0$, and Lemma~5.2 shows that
$P,BPB^\dagger,B^2PB^{-2}$ commute pairwise. Equation~(7.5)
shows that $J$ is real, and Lemma~7.3 gives
$A=JBJ^\dagger B^\dagger$. Corollary~4.6 therefore gives
$U\in\mathcal C_5(5)$. Since
$\mathsf N^4\boldsymbol p=\boldsymbol e_4\ne0$, Lemma~5.7 gives
$U\notin\mathcal C_4(5)$. Finally, Lemmas~5.5--5.6 show that
$\mathsf A$ and $\mathsf B$ have no common invariant Lagrangian;
since the conditional matrices are
$\mathbb I_{V_3},\mathsf A,\mathsf B,\mathsf B\mathsf A$,
Theorem~2.5 therefore shows that $U$ is not GSC.
\end{proof}

The counterexample extends to any larger number of qubits.
For every $r\ge1$ and $V\in\mathrm U(\HH_5)$,
\[
\Id{\HH_r}\otimes V\in\CC_k(r+5)
\quad\Longleftrightarrow\quad
V\in\CC_k(5)
\qquad(k\ge1).
\]
This follows by induction from the Pauli level, using
Lemma~\ref{lem:commutator-test} and the identity
$\Delta_{Q\otimes P}(\Id{\HH_r}\otimes V)
=\Id{\HH_r}\otimes\Delta_PV$
for $Q\in\PP_r$ and $P\in\PP_5$.
Thus $W:=\Id{\HH_r}\otimes U$ belongs to
$\CC_5(r+5)\setminus\CC_4(r+5)$.

If $W$ were GSC, there would be maximal commuting Pauli
algebras $\Alg,\Alg'$ with $W\Alg W^\dagger=\Alg'$.
Compressing on the idle register gives
\[
U\bigl(\bra{0^r}\Alg\ket{0^r}\bigr)U^\dagger
=\bra{0^r}\Alg'\ket{0^r}.
\]
Both compressed spaces are maximal commuting Pauli algebras
by Lemma~\ref{lem:compression}, so $U$ would be GSC,
contrary to Theorem~\ref{thm:counterexample}.
Together with the original five-qubit case and nestedness,
this gives non-GSC gates in $\CC_k(n)$ for every
$n\ge5$ and $k\ge5$.

Gottesman--Mochon exhibited a third-level gate whose inverse
is not in the third level, as reported by Beigi--Shor~\cite[Section~1.1]{BS}.  Recently, He--Robitaille--Tan~\cite[Theorem~5.4]{HRT}
constructed, for every $k\ge3$, a third-level gate $U_k$
whose inverse is outside $\CC_k$. We now show that the
inverse of our fixed five-qubit counterexample belongs
to no finite level of the Clifford hierarchy.

\begin{theorem}[The Clifford hierarchy is not closed under inverses]\label{prop:inverse-outside-hierarchy}
The gate
\[
U=\operatorname{ctrl}_{\mathsf c}(B)
  \operatorname{ctrl}_{\mathsf d}(A)
\]
of Theorem~8.1 satisfies
\[
U^\dagger\notin\bigcup_{k\ge1}\mathcal C_k(5).
\]
\end{theorem}

\begin{proof}
We will find a sequence of Pauli commutators whose target symplectic
action never becomes the identity. This excludes every finite hierarchy
level, since each Pauli commutator lowers the level.

Set
\[
T=\Delta_{X_{\mathsf d}}(\Delta_{X_{\mathsf c}}(U^\dagger)).
\]
The first commutator compares the two choices of \(B\); the second
compares the resulting operations with their conjugates by \(A\).
Explicitly, Lemma \ref{lem:all-commutators} gives the target operators
\[
K_0=A^{-1}B^{-1}AB,
\qquad
K_1=A^{-1}BAB^{-1}
\]
at control values
\((z_{\mathsf c},z_{\mathsf d})=(0,0)\) and \((1,0)\), respectively.

We now compare these two operators by putting
\[
R=K_1K_0^{-1}.
\]
The target operators of \(\Delta_{X_{\mathsf c}}T\) at the same control
values are \(R\) and \(R^{-1}\). Applying the same commutator again
exchanges these operators and multiplies each by the inverse of its
previous value. Thus each application gives
\[
(R^e,R^{-e})\longmapsto(R^{-2e},R^{2e}),
\]
so the target operator at \((0,0)\) after \(m+1\) such applications is
\[
R^{(-2)^m}.
\]

The same repeated-squaring observation appears in the proof
of Anderson and Weippert~\cite[Theorem~2.4]{AW}.

If \(U^\dagger\) belonged to a finite hierarchy level, Lemma~\ref{lem:commutator-test} would
force these commutators eventually to give a Pauli up to phase. Each of these commutators remains block diagonal on the full controls. A full Pauli with this property has only $\mathbb I$- or
$Z$-components on those controls, so each conditional target
operator is a phase times a target Pauli. Since the $(0,0)$
target block is $R^{(-2)^m}$, if \(U^\dagger\) belonged to a finite hierarchy level, then
\begin{equation}\label{eq:inverse-two-power}
\mathsf R^{2^m}=\mathbb I_{V_3}
\qquad\text{for some }m\ge0.
\end{equation}
The obstruction will be an action of order three: taking powers of two
cannot eliminate it.

To calculate \(\mathsf R\), use the Pauli rotations defining \(A\). Put
\[
\boldsymbol p_j=[B^jY_3B^{-j}],
\qquad
\mathsf J_j\boldsymbol v
=
\boldsymbol v+
\omega(\boldsymbol p_j,\boldsymbol v)\boldsymbol p_j.
\]

By Lemma~\ref{lem:rotation}, each \(\mathsf J_j\) is an involution, and two such
matrices commute when their defining vectors are orthogonal. Moreover,
\begin{equation}\label{eq:inverse-rotation-identities}
\mathsf A^{-1}=\mathsf A=\mathsf J_0\mathsf J_1,
\qquad
\mathsf B\mathsf J_j\mathsf B^{-1}=\mathsf J_{j+1}.
\end{equation}

The Pauli propagation rules for $\CNOT, \CZ,$ and $H$ give
\[
\begin{aligned}
\boldsymbol p_{-1}&=[X_1Z_2Y_3],
&
\boldsymbol p_0&=[Y_3],
\\
\boldsymbol p_1&=[Z_2Y_3],
&
\boldsymbol p_2&=[Z_1Y_3].
\end{aligned}
\]
Among these vectors, the only nonzero pairing is
\[
\omega(\boldsymbol p_{-1},\boldsymbol p_2)=1.
\]
Indeed, the corresponding Paulis anticommute on target qubit \(1\);
every other pair commutes. Thus almost all the factors introduced
by \eqref{eq:inverse-rotation-identities} commute, allowing the
following cancellations:
\[
\begin{aligned}
\mathsf K_0
&=\mathsf A\mathsf B^{-1}\mathsf A\mathsf B
 =\mathsf J_0\mathsf J_1\mathsf J_{-1}\mathsf J_0
 =\mathsf J_1\mathsf J_{-1},
\\
\mathsf K_1
&=\mathsf A\mathsf B\mathsf A\mathsf B^{-1}
 =\mathsf J_0\mathsf J_1\mathsf J_1\mathsf J_2
 =\mathsf J_0\mathsf J_2,
\\
\mathsf R
&=\mathsf K_1\mathsf K_0^{-1}
 =\mathsf J_0\mathsf J_1\mathsf J_2\mathsf J_{-1}.
\end{aligned}
\]

Consider the two-dimensional space
\[
E=\operatorname{span}_{\Ztwo}\{\boldsymbol p_2,\boldsymbol p_{-1}\}.
\]
The matrices \(\mathsf J_0,\mathsf J_1\) fix \(E\) pointwise because
their defining vectors are orthogonal to \(E\). The remaining two
factors preserve \(E\), and their action is determined by
\(\omega(\boldsymbol p_2,\boldsymbol p_{-1})=1\). In the ordered basis
\((\boldsymbol p_2,\boldsymbol p_{-1})\),
\[
\left.\mathsf R\right|_E
=
\left.(\mathsf J_2\mathsf J_{-1})\right|_E
=
\begin{pmatrix}
0&1\\
1&1
\end{pmatrix}.
\]
This matrix cyclically permutes the three nonzero vectors of \(E\):
\[
\boldsymbol p_2
\longmapsto
\boldsymbol p_{-1}
\longmapsto
\boldsymbol p_2+\boldsymbol p_{-1}
\longmapsto
\boldsymbol p_2.
\]
The order of $\sR|_E$ is therefore three, so no power \(2^m\) makes it the
identity. This contradicts \eqref{eq:inverse-two-power}.
\end{proof}

In forthcoming work, we will present qudit counterexamples.

\section*{Use of large language models}\label{sec:AI}

Commercially available large language models were used during the development of this research. This work builds on our prior research, which was further developed through LLM-assisted mathematical discussion over an extended period using our notes, ideas, and directions. In one such interaction, prompted using our prior notes and mathematical guidance, an LLM proposed an initial counterexample. We subsequently analysed, simplified, and reorganised the construction through further mathematical work and LLM-assisted discussion, leading to the simpler counterexample presented here. LLMs were also used in drafting and editing this preprint. The authors have verified all mathematical arguments and take full responsibility for the correctness and exposition of the work.

\section*{Acknowledgements}

ND acknowledges support from the Canada Research Chairs program and NSERC Discovery Grant RGPIN-2022-03103. OL acknowledges support from the Department of Pure Mathematics, University of Waterloo and the NSERC Canada Graduate Scholarship - Master's.

\clearpage
\appendix

\section{Calculations for three commuting Pauli conjugates}

\label{app:commuting-rotations}
We prove Lemma~\ref{lem:three-commuting}. For all $j\in\{0,1,2\}$, let $P_j:=B^jPB^{-j}\in\PP_n$ and $J_j:=(\Id{\HH_n}-iP_j)/\sqrt2\in\Cliff_n$. Conjugating the polynomial defining $J$ gives $J_j=B^jJB^{-j}$. Therefore
\begin{align*}
A&=JBJ^\dagger B^\dagger=J_0J_1^\dagger,\\
BAB^\dagger&=(BJ_0B^\dagger)(BJ_1^\dagger B^\dagger)=J_1J_2^\dagger.
\end{align*}

Since $P_0,P_1,P_2$ commute, their rotations $J_0,J_1,J_2$ commute by Lemma~\ref{lem:rotation}(iii). The corresponding symplectic matrices also commute, and each squares to identity. Thus
\[
\sA=\sJ_0\sJ_1,\qquad
\sB\sA\sB^{-1}=\sJ_1\sJ_2.
\]
Both matrices square to identity; for example,
\[
\sA^2=\sJ_0\sJ_1\sJ_0\sJ_1=\sJ_0^2\sJ_1^2=\Id{V_n}.
\]
Their commutation follows from commutation of the three factors. More precisely,
\begin{align*}
\sB\sA\sB^{-1}\sA
&=\sJ_1\sJ_2\sJ_0\sJ_1\\
&=\sJ_0\sJ_2\sJ_1^2=\sJ_0\sJ_2.
\end{align*}
The last product also squares to identity. Hence
\[
(\sB\sA\sB^{-1}\sA-\Id{V_n})^2
=(\sJ_0\sJ_2)^2+\Id{V_n}=0.
\]
This proves all parts of Lemma~\ref{lem:three-commuting}. The cancellations $\sJ_j^2=\Id{V_n}$ take place in the binary symplectic representation; the exact operator identity is $J_j^2=-iP_j$, not $J_j^2=\Id{\HH_n}$.

\section{Solving the alternating-form equations}
\label{app:form-calculations}
We give the elimination and nonsingularity check used in Lemma~\ref{lem:required-form}. The standard symplectic form is nondegenerate, so its matrix in a basis is nonzero.

Symmetry and zero diagonal leave the fifteen entries $g_{ij}$ with $0\leq i<j\leq5$. The equations with $i=j$ in \eqref{eq:form-equations} vanish identically, and those with $i>j$ repeat the equations with $i<j$. The complete list of remaining equations is displayed below; every expression in the second and fourth columns is equal to zero.
\begin{center}
\renewcommand{\arraystretch}{1.24}
\begin{tabular}{@{}c l @{\hspace{12mm}} c l@{}}
\toprule
$(i,j)$ & Expression & $(i,j)$ & Expression\\
\midrule
$(0,1)$ & $g_{02}+g_{12}$ & $(1,5)$ & $g_{25}$\\
$(0,2)$ & $g_{03}+g_{12}+g_{13}$ & $(2,3)$ & $g_{24}+g_{34}$\\
$(0,3)$ & $g_{04}+g_{13}+g_{14}$ & $(2,4)$ & $g_{25}+g_{34}+g_{35}$\\
$(0,4)$ & $g_{05}+g_{14}+g_{15}$ & $(2,5)$ & $g_{35}$\\
$(0,5)$ & $g_{15}$ & $(3,4)$ & $g_{35}+g_{45}$\\
$(1,2)$ & $g_{13}+g_{23}$ & $(3,5)$ & $g_{45}$\\
$(1,3)$ & $g_{14}+g_{23}+g_{24}$ & $(4,5)$ & $0$\\
$(1,4)$ & $g_{15}+g_{24}+g_{25}$ & &\\
\bottomrule
\end{tabular}
\end{center}
The equations with second index $5$ give
\[
g_{15}=g_{25}=g_{35}=g_{45}=0.
\]
The equation at $(1,4)$ then gives $g_{24}=0$, and the one at $(2,3)$ gives $g_{34}=0$. Since $g_{02}=0$ is required, the equation at $(0,1)$ gives $g_{12}=0$.

The equations at $(0,2)$, $(1,2)$, and $(1,3)$ now imply
\[
g_{03}=g_{13}=g_{23}=g_{14}.
\]
The equation at $(0,3)$ gives $g_{04}=g_{13}+g_{14}=0$, and the equation at $(0,4)$ gives $g_{05}=g_{14}$. Together with $g_{01}=g_{02}=0$, these determine every entry: the only possibly nonzero entries above the diagonal are
\[
g_{03},\quad g_{05},\quad g_{13},\quad g_{14},\quad g_{23},
\]
and they are all equal. If they were zero, the entire form would be zero. Nondegeneracy therefore forces all five to equal $1$, giving precisely $\Gstar$ in \eqref{eq:Gstar}. Its entries satisfy every equation in the table, so a basis with these pairings makes the prescribed matrix $\sB$ symplectic.

To check nonsingularity, suppose
\[
\Gstar(\alpha_0,\alpha_1,\alpha_2,\alpha_3,\alpha_4,\alpha_5)^\top=\bm0.
\]
Rows $5,4,2$ give $\alpha_0=0$, $\alpha_1=0$, and $\alpha_3=0$, respectively. Row $1$ then gives $\alpha_4=0$, row $0$ gives $\alpha_5=0$, and row $3$ gives $\alpha_2=0$. Thus the kernel is zero. This proves that $\Gstar$ is the unique nonzero solution of the required pairing equations and is nonsingular.

\section{Exact Pauli and gate calculations}
\label{app:gate-calculations}
We supply the calculations for Lemmas~\ref{lem:standard-classes}--\ref{lem:exact-A} and identity~\eqref{eq:exact-J}. 

\subsection{From the cyclic basis to standard Pauli classes}
\label{app:standard-classes}
The six standard coordinate classes can be recovered from the basis in \eqref{eq:pauli-basis} as
\begin{equation}
\begin{aligned}
[X_{\mathtt{1}}]&=\ve_3+\ve_2+\ve_4+\ve_5, & [Z_{\mathtt{1}}]&=\ve_2,\\
[X_{\mathtt{2}}]&=\ve_4, & [Z_{\mathtt{2}}]&=\ve_1,\\
[X_{\mathtt{3}}]&=\ve_5, & [Z_{\mathtt{3}}]&=\ve_0+\ve_5.
\end{aligned}
\label{eq:inverse-basis}
\end{equation}

Assume the cyclic-basis action in Lemma~\ref{lem:standard-classes}(i). Applying $\sB$ to \eqref{eq:inverse-basis} gives
\begin{align*}
\sB[X_{\mathtt{1}}]
&=(\ve_3+\ve_4)+(\ve_2+\ve_3)+(\ve_4+\ve_5)+\ve_5\\
&=\ve_2=[Z_{\mathtt{1}}],\\[2pt]
\sB[Z_{\mathtt{1}}]&=\ve_2+\ve_3=[Z_{\mathtt{1}}]+[Y_{\mathtt{1}}X_{\mathtt{2}}X_{\mathtt{3}}]=[X_{\mathtt{1}}X_{\mathtt{2}}X_{\mathtt{3}}],\\[2pt]
\sB[X_{\mathtt{2}}]&=\ve_4+\ve_5=[X_{\mathtt{2}}X_{\mathtt{3}}],\\
\sB[Z_{\mathtt{2}}]&=\ve_1+\ve_2=[Z_{\mathtt{1}}Z_{\mathtt{2}}],\\[2pt]
\sB[X_{\mathtt{3}}]&=\ve_5=[X_{\mathtt{3}}],\\
\sB[Z_{\mathtt{3}}]&=(\ve_0+\ve_1)+\ve_5=[Z_{\mathtt{2}}Z_{\mathtt{3}}].
\end{align*}
These prove (ii). Conversely, each of (i) and (ii) prescribes a unique linear map on a basis of $V_3$. Since the map specified by (i) has the images in (ii), a map satisfying (ii) must be that same map and therefore satisfies (i).

\subsection{The real circuit and its exact generator images}
\label{app:exact-B}
The standard Hadamard and CNOT identities are
\begin{align*}
H_{\mathtt{j}}X_{\mathtt{j}}H_{\mathtt{j}}&=Z_{\mathtt{j}},& H_{\mathtt{j}}Z_{\mathtt{j}}H_{\mathtt{j}}&=X_{\mathtt{j}},& H_{\mathtt{j}}Y_{\mathtt{j}}H_{\mathtt{j}}&=-Y_{\mathtt{j}},\\
\CNOT_{\mathtt{j}\to\mathtt{k}}X_{\mathtt{j}}\CNOT_{\mathtt{j}\to\mathtt{k}}&=X_{\mathtt{j}}X_{\mathtt{k}},&
\CNOT_{\mathtt{j}\to\mathtt{k}}X_{\mathtt{k}}\CNOT_{\mathtt{j}\to\mathtt{k}}&=X_{\mathtt{k}},\\
\CNOT_{\mathtt{j}\to\mathtt{k}}Z_{\mathtt{j}}\CNOT_{\mathtt{j}\to\mathtt{k}}&=Z_{\mathtt{j}},&
\CNOT_{\mathtt{j}\to\mathtt{k}}Z_{\mathtt{k}}\CNOT_{\mathtt{j}\to\mathtt{k}}&=Z_{\mathtt{j}}Z_{\mathtt{k}}.
\end{align*}
The Hadamard identities follow by multiplying
\[
H=\frac1{\sqrt2}\begin{pmatrix}1&1\\1&-1\end{pmatrix},\qquad
X=\begin{pmatrix}0&1\\1&0\end{pmatrix},\qquad
Z=\begin{pmatrix}1&0\\0&-1\end{pmatrix},
\]
and using $Y=iXZ$. For CNOT, put $D=\CNOT_{\mathtt{j}\to\mathtt{k}}$. The identity $D\ket{a,b}=\ket{a,a+b}$ gives $D=D^\dagger$ and
\begin{align*}
DX_{\mathtt{j}}D\ket{a,b}&=\ket{a+1,b+1},&
DX_{\mathtt{k}}D\ket{a,b}&=\ket{a,b+1},\\
DZ_{\mathtt{j}}D\ket{a,b}&=(-1)^a\ket{a,b},&
DZ_{\mathtt{k}}D\ket{a,b}&=(-1)^{a+b}\ket{a,b}.
\end{align*}
These are exactly the actions asserted above. Both elementary gates are real, and their Pauli conjugation identities show that they are Clifford.

For $B=\CNOT_{\mathtt{2}\to\mathtt{3}}\CNOT_{\mathtt{1}\to\mathtt{2}}H_{\mathtt{1}}$, apply these identities in order, starting with $H_{\mathtt{1}}$:
\begin{center}
\renewcommand{\arraystretch}{1.3}
\begin{tabular}{@{}c c c c@{}}
\toprule
Input & After $H_{\mathtt{1}}$ & After $\CNOT_{\mathtt{1}\to\mathtt{2}}$ & After $\CNOT_{\mathtt{2}\to\mathtt{3}}$\\
\midrule
$X_{\mathtt{1}}$ & $Z_{\mathtt{1}}$ & $Z_{\mathtt{1}}$ & $Z_{\mathtt{1}}$\\
$Z_{\mathtt{1}}$ & $X_{\mathtt{1}}$ & $X_{\mathtt{1}}X_{\mathtt{2}}$ & $X_{\mathtt{1}}X_{\mathtt{2}}X_{\mathtt{3}}$\\
$X_{\mathtt{2}}$ & $X_{\mathtt{2}}$ & $X_{\mathtt{2}}$ & $X_{\mathtt{2}}X_{\mathtt{3}}$\\
$Z_{\mathtt{2}}$ & $Z_{\mathtt{2}}$ & $Z_{\mathtt{1}}Z_{\mathtt{2}}$ & $Z_{\mathtt{1}}Z_{\mathtt{2}}$\\
$X_{\mathtt{3}}$ & $X_{\mathtt{3}}$ & $X_{\mathtt{3}}$ & $X_{\mathtt{3}}$\\
$Z_{\mathtt{3}}$ & $Z_{\mathtt{3}}$ & $Z_{\mathtt{3}}$ & $Z_{\mathtt{2}}Z_{\mathtt{3}}$\\
\bottomrule
\end{tabular}
\end{center}
Every factor of $B$ is a real Clifford. The table proves all generator identities in Lemma~\ref{lem:exact-B}(ii). For $P=Y_{\mathtt{3}}$,
\begin{align*}
BY_{\mathtt{3}}B^\dagger
&=i(BX_{\mathtt{3}}B^\dagger)(BZ_{\mathtt{3}}B^\dagger)
=iX_{\mathtt{3}}(Z_{\mathtt{2}}Z_{\mathtt{3}})=Z_{\mathtt{2}}Y_{\mathtt{3}},\\
B^2Y_{\mathtt{3}}B^{-2}
&=B(Z_{\mathtt{2}}Y_{\mathtt{3}})B^\dagger
=(Z_{\mathtt{1}}Z_{\mathtt{2}})(Z_{\mathtt{2}}Y_{\mathtt{3}})=Z_{\mathtt{1}}Y_{\mathtt{3}}.
\end{align*}
This proves Lemma~\ref{lem:exact-B}(iii).

For all $0\leq j\leq5$, let $E_j\in\PP_3$ be the Hermitian representative displayed in the second column below, so that $[E_j]=\ve_j$. Their exact conjugations are
\begin{center}
\renewcommand{\arraystretch}{1.3}
\begin{tabular}{@{}c c c c c@{}}
\toprule
$j$ & $E_j$ & $BE_jB^\dagger$ & $\sB\ve_j$ & $\sN\ve_j$\\
\midrule
$0$ & $Y_{\mathtt{3}}$ & $Z_{\mathtt{2}}Y_{\mathtt{3}}$ & $\ve_0+\ve_1$ & $\ve_1$\\
$1$ & $Z_{\mathtt{2}}$ & $Z_{\mathtt{1}}Z_{\mathtt{2}}$ & $\ve_1+\ve_2$ & $\ve_2$\\
$2$ & $Z_{\mathtt{1}}$ & $X_{\mathtt{1}}X_{\mathtt{2}}X_{\mathtt{3}}$ & $\ve_2+\ve_3$ & $\ve_3$\\
$3$ & $Y_{\mathtt{1}}X_{\mathtt{2}}X_{\mathtt{3}}$ & $-Y_{\mathtt{1}}X_{\mathtt{3}}$ & $\ve_3+\ve_4$ & $\ve_4$\\
$4$ & $X_{\mathtt{2}}$ & $X_{\mathtt{2}}X_{\mathtt{3}}$ & $\ve_4+\ve_5$ & $\ve_5$\\
$5$ & $X_{\mathtt{3}}$ & $X_{\mathtt{3}}$ & $\ve_5$ & $\bm0$\\
\bottomrule
\end{tabular}
\end{center}
Only the $j=3$ row requires an additional sign calculation. Keeping the order $Y=iXZ$ gives
\[
BY_{\mathtt{1}}B^\dagger=iZ_{\mathtt{1}}(X_{\mathtt{1}}X_{\mathtt{2}}X_{\mathtt{3}})=-Y_{\mathtt{1}}X_{\mathtt{2}}X_{\mathtt{3}},
\]
so
\begin{align*}
B(Y_{\mathtt{1}}X_{\mathtt{2}}X_{\mathtt{3}})B^\dagger
&=(-Y_{\mathtt{1}}X_{\mathtt{2}}X_{\mathtt{3}})(X_{\mathtt{2}}X_{\mathtt{3}})X_{\mathtt{3}}\\
&=-Y_{\mathtt{1}}X_{\mathtt{3}}.
\end{align*}
Passing to Pauli classes then gives
\[
\sN\ve_2=[X_{\mathtt{1}}X_{\mathtt{2}}X_{\mathtt{3}}]+[Z_{\mathtt{1}}]=[-iY_{\mathtt{1}}X_{\mathtt{2}}X_{\mathtt{3}}]=\ve_3,
\]
\[
\sN\ve_3=[-Y_{\mathtt{1}}X_{\mathtt{3}}]+[Y_{\mathtt{1}}X_{\mathtt{2}}X_{\mathtt{3}}]=[-X_{\mathtt{2}}]=\ve_4.
\]

\subsection{The exact rotation about \texorpdfstring{$Y_{\mathtt{3}}$}{Y3}}
\label{app:exact-J}
With $Y=\begin{pmatrix}0&-i\\i&0\end{pmatrix}$, direct multiplication gives
\[
\frac{\Id{\Cplx^2}-iY}{\sqrt2}
=\frac1{\sqrt2}\begin{pmatrix}1&-1\\1&1\end{pmatrix}=HZ.
\]
Extending by identity on qubits $\mathtt{1},\mathtt{2}$ proves \eqref{eq:exact-J}, including its exact sign and reality.

\subsection{The exact commutator and its square}
\label{app:exact-A}
Let $J=H_{\mathtt{3}}Z_{\mathtt{3}}$. Since $J$ acts on qubit $\mathtt{3}$, it commutes with $H_{\mathtt{1}}$ and $\CNOT_{\mathtt{1}\to\mathtt{2}}$. Expanding $B^\dagger$ in reverse order therefore gives
\begin{align*}
A=JBJ^\dagger B^\dagger
&=J\CNOT_{\mathtt{2}\to\mathtt{3}}\CNOT_{\mathtt{1}\to\mathtt{2}}H_{\mathtt{1}}J^\dagger
 H_{\mathtt{1}}\CNOT_{\mathtt{1}\to\mathtt{2}}\CNOT_{\mathtt{2}\to\mathtt{3}}\\
&=J\CNOT_{\mathtt{2}\to\mathtt{3}}J^\dagger\CNOT_{\mathtt{2}\to\mathtt{3}}.
\end{align*}
For computational value $z\in\Ztwo$ on qubit $\mathtt{2}$, the action on qubit $\mathtt{3}$ is $JX^zJ^\dagger X^z$. It is identity when $z=0$. On the one-qubit target, the orientation $J=HZ$ gives
\[
JXJ^\dagger=HZXZH=-Z,\qquad JXJ^\dagger X=-ZX=XZ.
\]
Consequently, with identity on qubit $\mathtt{1}$ understood,
\begin{equation}
A=\proj0_{\mathtt{2}}\otimes \Id{\HH_{\mathtt{3}}}+\proj1_{\mathtt{2}}\otimes X_{\mathtt{3}}Z_{\mathtt{3}}
=\ctrl_{\mathtt{2}}(X_{\mathtt{3}}Z_{\mathtt{3}})
=\CNOT_{\mathtt{2}\to\mathtt{3}}\CZ_{\mathtt{2},\mathtt{3}}.
\label{eq:exact-A-conditional}
\end{equation}
The last equality uses the controlled-product identity.

Finally, $(XZ)^2=-\Id{\Cplx^2}$, so squaring \eqref{eq:exact-A-conditional} gives
\[
A^2=\proj0_{\mathtt{2}}\otimes \Id{\HH_{\mathtt{3}}}-\proj1_{\mathtt{2}}\otimes \Id{\HH_{\mathtt{3}}}=Z_{\mathtt{2}}.
\]
Since $Z_{\mathtt{2}}$ is Pauli, its corresponding symplectic matrix is identity, proving $\sA^2=\Id{V_3}$.

\end{document}